\documentclass{amsart}

\usepackage{graphicx}
\usepackage{latexsym,color}
\usepackage{graphicx}
\usepackage{amssymb}
 \usepackage{amsfonts}
 \usepackage{amsmath}
\usepackage{amsthm}
\usepackage[final]{changes}
\colorlet{Changes@Color}{orange}
\usepackage{subfig}
\usepackage{multirow}

\newtheorem{thm}{Theorem}[section]
 \newtheorem{dfn}[thm]{Definition}
 \newtheorem{lem}[thm]{Lemma}
 \newtheorem{rem}[thm]{Remark}
\newtheorem{cor}[thm]{Corollary}
\newtheorem{prop}[thm]{Proposition}
\newtheorem{asm}[thm]{Assumption}

\newtheorem{example}[thm]{Example}

\numberwithin{equation}{section}

\DeclareMathOperator*{\essinf}{ess\, inf}

\begin{document}

\title{Super--replication in Fully Incomplete Markets}
\thanks{*Partly supported by the
 Einstein Foundation Berlin Grant no.A 2012 137 and
Marie--Curie Career Integration Grant, no. 618235}
\thanks{$^+$Partly supported by the
Swiss National Foundation grant SNF 200021$\_$153555}

  \author{Yan Dolinsky* \address{
 Department of Statistics, Hebrew University of Jerusalem, Israel. \\
 e.mail: yan.dolinsky@mail.huji.ac.il}}${}$\\
\author{Ariel Neufeld$^+$ \address{
 Department of Mathematics, ETH Zurich, Switzerland.\\
 e.mail: ariel.neufeld@math.ethz.ch}
  ${}$\\
Hebrew University and ETH Zurich}

\date{\today}

\maketitle \markboth{Y.Dolinsky and A.Neufeld}{Fully Incomplete Markets}
\renewcommand{\theequation}{\arabic{section}.\arabic{equation}}
\pagenumbering{arabic}

\begin{abstract}
In this work, we introduce the notion of fully incomplete markets.
We prove that for these markets, the super--replication price coincides
with the model free super--replication price. Namely, the knowledge of the model does not reduce the
super--replication price.
We provide two families of fully incomplete models:
stochastic volatility models and rough volatility models.
Moreover, we give several computational examples.
Our approach is purely probabilistic\added{.} \deleted{and allows to deal with
general models of stochastic volatility.}
\end{abstract}
\vspace{10pt}

\section{Introduction}
We consider a financial market with one risky asset, which is modeled through a semi--martingale
defined on a filtered probability space.
We
introduce and study a new notion,
the notion of fully incomplete markets. Roughly speaking,
a fully incomplete market is a financial market
for which the set of absolutely continuous local martingale measures is dense
in a sense that will be explained formally in the sequel.
We prove that a wide range of stochastic volatility models (see for instance Heston (1993), Hull and White (1987)\deleted{model}
and Scott (1987)) and rough volatility models (see
Gatheral, Jaisson and Rosenbaum (2014)) are fully incomplete.

The main contribution of this work is the establishment
of a surprising link between super--replication in the model
free setup and in fully incomplete markets.
Namely, we prove that for fully incomplete markets, the knowledge of
the probabilistic model does not reduce the super--replication price, i.e.,
the classical super--replication price is equal to the
model free super--replication price.
We deal with two main setups of super--replication.
The first setup is the semi--static hedging of
European options and the second setup
is the super--replication of game options.

In the first setup, we assume that in addition to trading the stock, the investor is
allowed to take
static positions in a finite number of options (written on the underlying asset) with
initially known prices.
The financial motivation for this assumption is that vanilla options such as call options
are liquid and hence should be treated as
primary assets
whose prices are given in the market.

We consider the super--replication of bounded (path dependent) European options. Our main result in Theorem~\ref{thm2.1} says
that for fully incomplete markets, the super--replication price is the same as in the model free setup.
Moreover,
when the probabilistic model is given, we show in Theorem~\ref{thm.new} the novel result that there
is a hedge which minimizes the cost
of a super--replicating strategy, i.e., that there is
an optimal hedge. This is done by applying the Koml\'os compactness principle, see, e.g., Lemma A 1.1 in Delbaen and Schachermayer (1994).
This compactness principle requires an underlying probability space. Hence, in the continuous time model
free setup, the existence of an optimal hedge is an open question which is left for future research.

In Bouchard and Nutz (2015), the authors proved the existence of an optimal hedge in a general
quasi sure setup (which includes the model free setup). In  their non trivial proof, they first considered the one-period case and then extended it by induction to the multi-period case.
Clearly, such an approach is limited to the discrete-time setup.

The model-independent approach with semi-static hedging received considerable
attention in recent years. The first work in this direction is the seminal contribution
by Hobson (1998). For more recent results,
see for instance (Acciaio et al. (2015), Beiglboeck et al. (2015),
Dolinsky and Soner (2014, 2015(a)), Galichon, Henry-Labordere and Touzi (2014),
Guo, Tan and Touzi (2015), Hou and Ob{\l}{\'o}j (2015),
and Henry-Labordere et al. (2014)).

Our second setup deals with super--replication of game options.
A game contingent claim (GCC) or game option,
which was introduced in Kifer (2000),
is defined
as a contract between the seller and the buyer such
that both have the right to exercise it at any time up to a
maturity date (horizon) $T$. If the buyer exercises the contract
at time $t$, then he receives the payment $Y_t$, but if the seller
exercises (cancels) the contract before the buyer, then the latter
receives $X_t$. The difference $\Delta_t=X_t-Y_t\added{\geq0}$ is the penalty
that the seller pays to the buyer for the contract cancellation.

A hedging strategy against a GCC is defined as a pair $(\pi,\sigma)$, which consists
of a self financing portfolio $\pi$ and a stopping time $\sigma$ representing the cancellation time
for the seller. A hedging strategy
is super-replicating the game option if no matter what exercise
time the buyer chooses, the seller can cover his liability to the
buyer (with probability one).
The super--replication price $V^*$ is defined as
the minimal initial capital which is required for a super-replicating strategy,
i.e., for any $\Xi>V^*$ there is a super-replicating strategy with an initial
capital $\Xi$.

For the above two setups (semi--static hedging of European options and
hedging of game options),
we prove that for fully incomplete markets, the super--replication price is the cheapest cost of a
trivial super--replicating strategy and coincides with the model free super--replication price.
For game options, a trivial hedging strategy is a pair which consists of a buy--and--hold portfolio
and a hitting time of the stock price process.
We show that for path independent payoffs
$X_t=f_2(S_t)$ and $Y_t=f_1(S_t)$, the super--replication price
equals $g(S_0)$ where $g$ (determined by $f_1,f_2$) can be viewed as
the game variant of a concave envelop\added{e}.
We give a characterization of the optimal hedging strategy
and provide several examples for explicit calculations of the above.

We note that the above two setups were studied recently for the case where hedging of the stock is subject
to proportional transaction costs (see Dolinsky (2013) for the game options setup and
Dolinsky and Soner (2015b) for semi--static hedging of European options). In these two papers,
it was shown that if the logarithm of the discounted stock price process satisfies the
conditional full support property (CFS) then
the super--replication price coincides with the model free super--replication price.
Thus, our results in the present paper show that the behavior of super--replication
prices in fully incomplete markets (without transaction costs)
is similar to their behavior in the presence of proportional transaction costs
in  markets which satisfy the CFS property.
Intuitively, one might expect that the notion of fully incomplete market is stronger than the CFS property.
However, as we will see in Remark \ref{rem.compare}, these two properties are in general not comparable.

In Cvitanic, Pham and Touzi (1999),
the authors studied the super--replication of European options
in the \replaced{presence}{present} of portfolio constraints and
stochastic volatility.
One of their results says that if the stochastic
volatility is unbounded (and satisfies some continuity assumptions),
then, even in the unconstrained case, the super--replication price
is the cheapest cost of a  buy--and--hold super-replicating portfolio,
and is given in terms of the concave envelop\added{e} of the payoff.
These results can trivially \added{be} extended
to the case of American options.
The main tool that the authors used relies on a
PDE approach to control theory of Markov processes (Bellman equation).

Our results are an extension of the results in
Cvitanic, Pham and Touzi (1999).
We present a purely probabilistic approach, which is
based on a change of measure.
The main idea of our approach is that,
in a sufficiently rich probability space,
the set of the distributions of the discounted stock price process
under equivalent martingale measures is dense in the set of all martingale measures.
We give an exact meaning
to this statement in Lemma \ref{lem.density}.

The idea to use a change of measure
for the construction
of dense pricing distributions goes back to Kusuoka (1992).
In this unpublished working paper,
Kusuoka deals with super--replication prices of European
 options in the Black--Scholes model in the presence of proportional transaction costs.
The author uses the Girsanov theorem in order to construct a set of shadow prices
such that any Brownian martingale (with some regularity assumptions) is a cluster point of this set.

Several important questions remain open and are left for future research.
The first question is whether our results can be extended to a more general setup
of super--replication, where we super--replicate
American or game options and permit static positions in
European and American options. Recently, several papers studied
static hedging of American options (with European options/American options) in a discrete-time setting,
see Bayraktar, Huang and Zhou (2015), Bayraktar and Zhou (2015, 2016), Deng and Tan (2016), and Hobson and Neuberger (2016).
The second question is whether one can extend the results to the case of multiple risky assets.
It seems that our definition for fully incomplete market can be extended
to this case as well. But in this instance,
it is not clear what the game variant of a concave envelop\added{e}
and the cheapest cost of a
trivial super--replicating strategy are.
We leave the
technicalities
for future research.
Another task is to provide an
interesting computational example
for model free semi--static hedging with finitely many options.
This was not done so far, even for the case of one risky asset.
We remark on more open questions in
Sections~\ref{sec:3+}~and~\ref{sec:proof}.

The paper is organized as follows. In the next section, we introduce the
concept of fully incomplete markets and
argue that a wide range of stochastic volatility models and
rough volatility models are fully incomplete. This is proven in Section~\ref{sec:3}.
In Section~\ref{sec:3+}, we formulate and prove our main results
for semi--static hedging of European options.
In Section~\ref{sec:3++}, we formulate our main results for game options. Furthermore,
we provide several examples for which we calculate explicitly
the super--replication price and the corresponding optimal hedging strategy.
In Section~\ref{sec:proof}, we prove our results for game options.
To that end, we prove some auxiliary lemmas in Section~\ref{sec:lemmas}. In the last
Section, we give an exact meaning to
the density property of fully incomplete markets.

\section{Fully Incomplete Markets}\label{sec:2}\setcounter{equation}{0}
Let $T$ be a finite time horizon and let $(\Omega,\mathcal F,\{\mathcal F_t\}_{t=0}^T,\mathbb P)$
be a complete probability space endowed with a
filtration $\{\mathcal F_t\}_{t=0}^T$ satisfying the usual conditions.
We consider a financial market which consists
of a \replaced{savings account}{non risky asset} $B={\{B_t\}}_{t=0}^T$
and of a stock $S={\{S_t\}}_{t=0}^T$. The \replaced{savings account}{non risky asset} is given by
\begin{equation}\label{2.1}
dB_t=r_t B_t\, dt\added{,}\quad B_0=1,
\end{equation}
where ${\{r_t\}}_{t=0}^T$
is a non--negative adapted stochastic process which represents the interest rate. We will assume that
${\{r_t\}}_{t=0}^T$ is uniformly bounded.
The risky asset is given by
\begin{equation}\label{2.2}
d S_t=S_t\left(r_t\,dt+\nu_t\,dW_t\right), \ \ S_0>0,
\end{equation}
where $\nu=\{\nu_t\}_{t=0}^T$ is a progressively measurable process with given starting point $\nu_0>0$ satisfying $\int_0^T \nu^2_s\,ds<\infty \ \mathbb{P}$-a.s., and where
$W=\{W_t\}_{t=0}^T$ is a Brownian motion with respect to the filtration
$\{\mathcal F_t\}_{t=0}^T$.

Let
$\mathcal{C}(\nu_0)$ be the set of all
continuous, strictly positive
stochastic processes $\alpha={\{\alpha_t\}}_{t=0}^T$
which are adapted with respect to the filtration generated by $W$ completed by the null sets, and satisfy:
i. $\alpha_0=\nu_0$. ii.
$\alpha$ and $\frac{1}{\alpha}$ are uniformly bounded.
\begin{dfn}\label{dfn2.1}
A financial market given by (\ref{2.1})--(\ref{2.2}) is
called fully incomplete if
for any $\epsilon>0$
and any process $\alpha\in \mathcal C(\nu_0)$ there exists a
probability measure $\mathbb Q\ll\mathbb P$ such that:\\
i. ${\{W_t\}}_{t=0}^T$ is a \replaced{Brownian motion}{Wiener process}
with respect to the probability measure $\mathbb Q$ and the filtration
${\{\mathcal{F}_t\}}_{t=0}^T$.\\
ii.
\begin{equation}\label{dnew}
\mathbb Q(\|\alpha-\nu\|_{\infty}>\epsilon)<\epsilon,
\end{equation}
where $\|u-v\|_{\infty}:=\sup_{0\leq t\leq T}|u_t-v_t|$ is the distance between $u$ and $v$ with respect to the
uniform norm.
\end{dfn}
Let us briefly explain the intuition behind the definition of a fully incomplete market.
Consider the discounted stock price
$\tilde S_t:=\frac{S_t}{B_t}$, $t\in [0,T]$. From (\ref{2.1})--(\ref{2.2}), we get
 $d\tilde S_t=\nu_t\tilde S_t\, dW_t$. Thus, Definition \ref{dfn2.1}
 says that for a fully incomplete market,
for any volatility process $\alpha\in C(\nu_0)$,
 we can find an absolutely continuous
 local martingale measure $\mathbb Q\ll\mathbb P$
 under which the volatility of the discounted
 stock price $\tilde S$ is close to $\alpha$.
In fact, using density arguments, we will see (in Lemma~\ref{lem.density}) that in fully incomplete markets, the set of the
distributions of the discounted stock price under absolutely continuous
local martingale measures is dense in the set of all local martingale distributions.
 \begin{rem}
Observe that the probability measure
$\mathbb P$ is already a
local martingale measure. Thus, by taking convex combinations of the from
$\lambda \mathbb P+(1-\lambda)\mathbb Q$
where $\lambda>0$ is "small" and $\mathbb Q$ is an absolutely continuous
local martingale measure, we deduce the following. If Definition \ref{dfn2.1} is satisfied, then
if we change the condition $\mathbb Q\ll \mathbb P$ to the more restrictive
condition
$\mathbb Q\sim \mathbb P$ of equivalent probability measures, the modified definition will be satisfied as well.
\end{rem}
The main results of this paper (which are formulated in Sections \ref{sec:3+}--\ref{sec:3++})
say that for fully incomplete markets the super--replication price
is the same as for the path--wise model free setup.
Namely, the knowledge of the probabilistic model does not
reduce the super--replication price.
We will formulate and prove this result for two setups.
The first setup is a semi--static European options' hedging model.
The second setup deals with game options.

The following Proposition (which will be proved in Section \ref{sec:3})
provides two families of stochastic volatility models which are fully incomplete.
\begin{prop}\label{prop2.1}
${}$\\
I.
Consider the following stochastic volatility model:
\begin{equation}\label{2.3}
d\nu_t=a(t,\nu_t)\,dt+b(t,\nu_t) \,d{\hat W}_t+c(t,\nu_t) \,dW_t, \ \ \nu_0>0,
\end{equation}
where $\hat W=\{\hat W_t\}_{t=0}^T$ is a Brownian motion with respect to
$\{\mathcal F_t\}_{t=0}^T$ which is independent
of $W$.
Assume that the SDE (\ref{2.3}) has a unique strong solution and the solution is strictly positive.
If
the functions $a,b,c: [0,T]\times (0,\infty)\rightarrow\mathbb R$
are continuous and for any $t\in [0,T]$, $x>0$ we have $b(t,x)>0$, then
the financial market given by (\ref{2.1})--(\ref{2.2}) is fully incomplete. \\
II. Let $\{\mathcal F_t\}_{t=0}^T$
be the usual augmentation of the filtration generated by $W$ and $\nu$.
Assume a decomposition
$\nu_t=\nu^{(1)}_t\nu^{(2)}_t$ where
$\nu^{(1)}$ is adapted to the filtration generated by $W$, and
$\nu^{(2)}$ is independent of $W$.
Moreover, assume that
$\nu^{(1)},\nu^{(2)}$ are strictly positive and continuous processes.
If $\ln\nu^{(2)}$ has a conditional full support (CFS) property,
then
the market given by (\ref{2.1})--(\ref{2.2}) is fully incomplete.
\end{prop}
Recall that a stochastic process $\Sigma=\{\Sigma_t\}_{t=0}^T$ has the
CFS property if
for all $t\in (0,T]$
$$\mbox{supp} \ \mathbb P(\Sigma_{|[t,T]}|\Sigma_{|[0,t]})=C_{\Sigma_t}[t,T] \ \ \mbox{a.s.,}$$
where
$C_y[t,T]$ is the space of all continuous functions $f:[t,T]\rightarrow\mathbb{R}_{+}$
with $f(t)=y$. In words, the CFS property prescribes that
from any given time on, the asset price path can continue arbitrarily close
to any given path with positive conditional probability.
\begin{rem}\label{rem.compare}
The notion of fully incomplete markets and the CFS property are in general not comparable.

It is well known that a Brownian motion with drift satisfies the CFS property.
Hence, e.g., the log price of the Black--Scholes model satisfy the CFS property, but being complete, it is clearly not fully incomplete.

Let us give a simple example of a fully incomplete market which does not satisfy the CFS property.
Consider a probability space which supports two independent Brownian motions $W$ and $\hat W$ and
a Bernoulli random variable $\xi\sim Ber (0.5)$ which is independent of $W$ and $\hat W$. Consider the market given by
(\ref{2.1})--(\ref{2.2}) with $r\equiv 0$ and
$\nu_t=e^{\hat W_t}\mathbb{I}_{\xi=0}$, $t\in [0,T]$. 
By looking at probability measures which are supported on the
event $\{\xi=0\}$, we deduce from Proposition \ref{prop2.1} 
(by applying any one of the two statements) that this market is fully incomplete.
On the other hand, we observe that it does not satisfy
the CFS property. Indeed, consider the event $D=\{S_t=S_0 \ \forall t\leq T/2\}$. Clearly, $D=\{\xi=1\}$. Hence,
$\mathbb P(D)=\mathbb P(\xi=1)=1/2$ and the conditional support
$\mbox{supp} \ \mathbb P(S_{|[T/2,T]}|D)$ contains only one function $f:[T/2,T]\rightarrow \mathbb R$ which is defined by
$f\equiv S_0$. This is a contradiction to the CFS property.

Even if we insist on strictly positive volatility, we can still construct similar examples
that produce martingales with atoms such that, with positive probability, the conditional support 
$\mbox{supp} \ \mathbb P(S_{|[T/2,T]}|S_{|[0,T/2]})$ is a finite set. This is clearly a contradiction
to the CFS property. 
Thus, without adding additional assumptions (it is an interesting question to understand what these assumptions would be),
full incompleteness in general does not imply the CFS property.
 \end{rem}
We end this section with several examples of fully incomplete markets.
\begin{example}
\textbf{Stochastic Volatility Models.}
\\
I. The Heston (1993) model:
\begin{eqnarray*}
&dS_t=S_t(r_t \,dt+ \sqrt\mathbf U_t \,dW^S_t)\\
&d\mathbf U_t=\kappa(\theta-\mathbf U_t) \,dt +\xi\sqrt\mathbf U_t \,dW^{\mathbf U}_t,
\end{eqnarray*}
 where $\{W^S_t\}_{t=0}^T$ and
 $\{W^{\mathbf U}_t\}_{t=0}^T$ are two Brownian motions with constant correlation
 $\rho\in (-1,1)$. Moreover,
 $\kappa,\theta,\xi>0$ are constants which satisfy
$2\kappa\theta>\xi^2$. The last condition guarantees that $\mathbf U$ is strictly positive.
Thus, applying It\^o's formula for $\nu_t:=\sqrt\mathbf U_t$ and
using the relations
$W^S=W$ and $W^{\mathbf U}=\rho W+\sqrt{1-\rho^2} \hat W$, we obtain that
$\nu$ is solution of (\ref{2.3}) with
$a(t,x)=\frac{\kappa}{2}\left(\frac{\theta}{x}-x\right)-\frac{\xi^2}{8 x}$,
$b(t,x)\equiv\frac{\xi}{2}\sqrt{1-\rho^2}$ and $c(t,x)\equiv\frac{\xi}{2}\rho$.\\
II. The Hull--White (1987) model:
\begin{eqnarray*}
&dS_t=S_t(r_t \,dt+ \sqrt\mathbf U_t \,dW^S_t)\\
&d\mathbf U_t=\mathbf U_t (\kappa  \,dt+\theta \,dW^{\mathbf U}_t),
\end{eqnarray*}
 where $\{W^S_t\}_{t=0}^T$ and
 $\{W^{\mathbf U}_t\}_{t=0}^T$ are two Brownian motions with constant correlation
 $\rho\in (-1,1)$ and $\kappa,\theta\in\mathbb R$ are constants.
Clearly, $\nu:=\sqrt\mathbf U$ satisfies the assumptions of Proposition \ref{prop2.1} (part I) . \\
III. The Scott (1987) model:
\begin{eqnarray*}
&dS_t=S_t( r \,dt +\lambda e^{\mathbf U_t} \,dW^S_t)\\
&d\mathbf U_t=-\kappa \mathbf U_t \,dt +\theta \,dW^{\mathbf U}_t,
\end{eqnarray*}
where $\{W^S_t\}_{t=0}^T$ and
 $\{W^{\mathbf U}_t\}_{t=0}^T$ are two Brownian motions with constant correlation
 $\rho\in (-1,1)$ and $\lambda,\kappa,\theta>0$ are constants. By applying It\^o's formula for
 $\nu:=\lambda e^{\mathbf U}$, this
model can be treated as the Heston model.
 \qed
 \end{example}
 \begin{example}
\textbf{Rough Volatility Models.}\\
Consider a model where the log--volatility is a fractional Ornstein--Uhlenbeck process (see Gatheral, Jaisson and Rosenbaum (2014)).
Formally, the volatility process is given by
$\nu_t=\nu_0 e^{\kappa U_t}$ where $\kappa>0$ is a constant and
$U_t=e^{-\lambda t}\int_{0}^t e^{\lambda u }dB^H_u$. Here,
$B^H=\{B^H_t\}_{t=0}^T$ is a fractional Brownian motion with Hurst parameter $H\in (0,1)$
and $\lambda>0$ is a constant.
The integral above is defined by integration by parts
$$\int_{0}^t e^{\lambda u}\,dB^H_u=e^{\lambda t}B^H_t-\lambda\int_{0}^t B^H_u e^{\lambda u}\,du.$$
Let $\{\mathcal F_t\}_{t=0}^T$
be the usual augmentation of the filtration generated by $W$ and $\nu$.

Assume that we have the representation
$B^H=\rho B^{H,1}+\sqrt{1-\rho^2} B^{H,2}$ where $\rho\in (-1,1)$ is a constant
and $B^{H,1},B^{H,2}$ are independent fractional Brownian motions. Moreover, assume that
$B^{H,1}$ is adapted to the filtration generated by $W$.
Then
$\nu_t=\nu^{(1)}_t\nu^{(2)}_t$ where
\begin{align*}
\nu^{(1)}_t&=\nu_0 \exp\left(\kappa \rho e^{-\lambda t}\int_{0}^t e^{\lambda u}\,dB^{H,1}_u\right),\\
\nu^{(2)}_t&=\exp\left(\kappa \sqrt{1-\rho^2} e^{-\lambda t}\int_{0}^t e^{\lambda u}\,dB^{H,2}_u\right).
\end{align*}
By Guasoni, Rasonyi and Schachermayer (2008, Proposition 4.2), fractional Brownian motion has the CFS property. This together with
Pakkanen (2010, Theorem 3.3) gives
that
$\ln\nu^{(2)}$ has the CFS property. Thus, as the assumptions of
the second statement in
Proposition \ref{prop2.1}
hold true, the market is fully incomplete.
\qed
\end{example}

\section{Semi--static Hedging}\label{sec:3+}\setcounter{equation}{0}
In this section, we deal with the super--replication of European options.
As the exercise time of the European options is fixed
(compared to game options), then for deterministic interest rates,
it is possible to discount the asset price and the payoffs
of the European options. Therefore, for that case, we can directly
assume without loss of generality that the interest rate is $r\equiv 0$.
For stochastic interest rate, writing the discounted payoffs of European options
in terms of the discounted asset price is not always possible,
and even when possible the new payoff function can loose its continuity.
Thus, in the case of stochastic interest rate, the assumption $r\equiv 0$ is not natural.
However, to make things simpler, we assume in this section that $r\equiv 0$.

Denote by $C[0,T]$ the space of all continuous functions
$f:[0,T]\rightarrow\mathbb R$ equipped with the uniform topology.
Consider a path-dependent European option with the payoff
$X=H(S)$, where $H:C[0,T]\rightarrow \mathbb{R}$ is a bounded and uniformly continuous function.
We assume that there are $N\geq 0$ static positions which can be bought at
time zero for a given price. Formally, the payoffs of the static positions
are given
by $X_i=h_i(S)$ where
$h_1,...,h_N:C[0,T]\rightarrow\mathbb R$ are bounded and uniformly continuous.
The price of the static position $X_i$ is denoted by $\mathcal P_i$.
Therefore, the initial stock price $S_0$ and the prices $\mathcal P_1,...,\mathcal P_N$
of the options $h_1,...,h_N$ are the data available in the market.

First, consider the case where the investor has probabilistic belief,
modeled by the given filtered probability space
$(\Omega,\mathcal F,\{\mathcal F_t\}_{t=0}^T,\mathbb P)$
introduced before. In this setup, a hedging strategy is a
pair $\pi=(c,\gamma)$ where $c=(c_0,...,c_N)\times\mathbb{R}^{N+1} $
and $\gamma={\{\gamma_t\}}_{t=0}^T$ is a progressively measurable process
with $\int_0^T \gamma_t^2 \nu^2_t S^2_t\,dt<\infty \ \mathbb P$-a.s., such that
the stochastic integral $\int \gamma dS$ is uniformly bounded from below. The corresponding
portfolio value at the maturity date is given by
 \begin{equation*}
 Z^\pi_T=c_0+\sum_{i=1}^N c_i h_i(S)+\int_{0}^T \gamma_u dS_u.
 \end{equation*}
The initial cost of the hedging strategy $\pi$ is
\begin{equation}\label{3+.2}
C(\pi)=c_0+\sum_{i=1}^N c_i\mathcal P_i.
\end{equation}
A strategy $\pi$ is a super--replicating strategy if
\begin{equation*}
Z^{\pi}_T\geq H(S) \ \ \mathbb P\mbox{-a.s.}
\end{equation*}
Then, the super--replication price is defined by
\begin{equation*}
V^{\mathbb P}_{h_1,...,h_N}(H)=\inf\{C(\pi): \ \pi \ \mbox{is} \ \mbox{a}
\ \mbox{super--replicating} \ \mbox{strategy}\}.
\end{equation*}

Next, consider the case where the investor has no probabilistic belief,
just the market data given as information.
Such an investor is modeled via the robust hedging approach.
Let $\{\mathbb S_t\}_{t=0}^T$ be the canonical process on the space $C[0,T]$, i.e.
$\mathbb S_t(\omega)=\omega(t)$, $\omega\in C[0,T]$.
Consider the corresponding
canonical filtration
$\mathbb F_t=\sigma\{\mathbb S_u: u\leq t\}$.
Denote by $\mathcal M$ the set of all probability measures
$Q$ on $C[0,T]$ such that under $Q$,
the process $\{\mathbb S_t\}_{t=0}^T$ is a strictly positive local martingale (with respect to its natural filtration) and
$\mathbb S_0=S_0$ $Q$-a.s.

In the robust setup, a hedging strategy is a pair $\pi=(c,\gamma)$ where $c\in \mathbb{R}^{N+1}$
and $\gamma={\{\gamma_t\}}_{t=0}^T$ is an adapted process ( w.r.t. the canonical filtration)
of bounded variation with left-continuous paths such that the process
$\int \gamma d\mathbb S$ is uniformly bounded from below, where here, we define
$$\int_{0}^T \gamma_u d\mathbb S_u:=\gamma_T\mathbb S_T-\gamma_0\mathbb S_0-\int_{0}^T \mathbb S_t d\gamma_t$$
using the standard Stieltjes integral for the last integral. The corresponding portfolio value at the maturity
date $T$ is given as before by
\begin{equation*}
\mathbb Z^{\pi}_T(\mathbb S)=c_0+\sum_{i=1}^N c_i h_i(\mathbb S)+\int_{0}^T \gamma_u d\mathbb S_u.
\end{equation*}

Moreover, as before, the cost of the hedging strategy
$\pi$ is given by (\ref{3+.2}).
The robust super--replication price is defined by
$$
  V_{h_1,...,h_N}(H)=\inf\{C(\pi): \exists  \pi \ \mbox{such} \ \mbox{that}
\ \mathbb Z^\pi_T(\mathbb S)\geq H(\mathbb S) \ \forall \mathbb S \ \mbox{strictly} \ \mbox{positive}, \ \mathbb S_0=S_0\}.
$$
The following theorem
says that
if
the financial market is fully incomplete, then
the corresponding super--replication price is the same as in the model
free setup. Namely, for fully incomplete markets the knowledge of the
probabilistic model does not reduce the super--replication price.
\begin{thm}\label{thm2.1}
Assume that the financial market given by
$\{S_t\}_{t=0}^T$ is fully incomplete. Then
$V^{\mathbb P}_{h_1,...,h_N}(H)=V_{h_1,...,h_N}(H).$ (might be $-\infty$).
\end{thm}
\begin{proof}
Clearly, $V^{\mathbb P}_{h_1,...,h_N}(H)\leq V_{h_1,...,h_N}(H)$, and so we need to establish the inequality
$
V^{\mathbb P}_{h_1,...,h_N}(H)\geq V_{h_1,...,h_N}(H).$

For a measurable function $\hat H:C[0,T]\rightarrow\mathbb R$
denote by $V^{\mathbb P}(\hat H)$ and $V(\hat H)$, the classical (i.e.  w.r.t. the probabilistic belief $\mathbb P$)
 and the robust super--replication price
of the claim $\hat H(S)$ for the case $N=0$, respectively.
Denote by $\mathcal Q$ the set of all probability measures
$\mathbb Q\ll \mathbb P$ such that $\{W_t\}_{t=0}^T$ is a
Brownian motion with respect to $\mathbb Q$ and
the filtration $\{\mathcal F_t\}_{t=0}^T$.

For any hedging strategy $\pi=(c,\gamma)$ and $\mathbb Q\in\mathcal Q$, the stochastic integral
$$\int_{0}^t \gamma_u dS_u=\int_{0}^t \gamma_u \nu_u S_udW_u, \ \ t\in [0,T]$$ is a local martingale bounded from below, hence a supermartingale. Thus, from Lemma \ref{lem.density} and the fact that
$H(S)-\sum_{i=1}^N c_i h_i(S)$ is a bounded and continuous function, we get
\begin{align*}
V^{\mathbb P}_{h_1,...,h_N}(H)&=\inf_{(c_1,...,c_N)\in\mathbb R^N}\bigg
(\sum_{i=1}^N c_i\mathcal P_i+V^{\mathbb P}\Big(H-\sum_{i=1}^N c_i h_i\Big)\bigg)\\
&\geq
\inf_{(c_1,...,c_N)\in\mathbb R^N}\bigg(\sum_{i=1}^N c_i\mathcal P_i+\sup_{\mathbb Q\in\mathcal Q}
\mathbb E_{\mathbb Q}[H(S)-\sum_{i=1}^N c_i h_i(S)]\bigg)\\
&\geq\inf_{(c_1,...,c_N)\in\mathbb R^N}
\bigg(\sum_{i=1}^N c_i\mathcal P_i+\sup_{Q\in\mathcal M}\mathbb
E_Q[H(\mathbb S)-\sum_{i=1}^N c_i h_i(\mathbb S)]\bigg).
\end{align*}
By applying
Hou and Ob{\l}{\'o}j (2015, Theorem~3.2) for the bounded and uniformly continuous claim
$H(S)-\sum_{i=1}^N c_i h_i(S)$, we obtain
\begin{align*}
& \ \inf_{(c_1,...,c_N)\in\mathbb R^N}
\bigg(\sum_{i=1}^N c_i\mathcal P_i+\sup_{Q\in\mathcal M}\mathbb E_Q[H(\mathbb S)-\sum_{i=1}^N c_i h_i(\mathbb S)]\bigg)\\
= & \ \inf_{(c_1,...,c_N)\in\mathbb R^N}\bigg
(\sum_{i=1}^N c_i\mathcal P_i+V\Big(H-\sum_{i=1}^N c_i h_i\Big)\bigg)\\
=& \ V_{h_1,...,h_N}(H)
\end{align*}
and the result follows.
\end{proof}
Next, we prove for the probabilistic model that there is an optimal super--replicating strategy,
i.e., a strategy which achieves the minimal cost.
To this end,
we need an additional assumption which rules out an arbitrage opportunity, i.e., a
case where $V^{\mathbb P}_{h_1,...,h_N}(H)=V_{h_1,...,h_N}(H)=-\infty$.
Thus, as in Hou and Ob{\l}{\'o}j (2015) (see Assumption 3.7 and Remark 3.8 there)
we assume the following.
\begin{asm}\label{asm.noarbitrage}
There is $\varepsilon>0$ such that for any
$(y_1,...,y_N)\in \prod_{i=1}^N [\mathcal P_i-\varepsilon,\mathcal P_i+\varepsilon]$
we can find a probability measure
$Q\in\mathcal M$
for which
$\mathbb E_{Q}[h_i(\mathbb S)]=y_i$, $i=1,...,N$.
\end{asm}
 \begin{thm}\label{thm.new}
Consider the super--replication problem on the filtered
probability space $(\Omega, \mathcal F,\{\mathcal F_t\}_{t=0}^T,\mathbb P)$ described above.
If Assumption \ref{asm.noarbitrage} holds true, then
there exists a super--replicating portfolio strategy $\hat\pi$ such
that $C(\hat\pi)=V^{\mathbb P}_{h_1,...,h_N}(H)$.
\end{thm}
\begin{proof}
Let $\pi^{(n)}=(c^{(n)},\gamma^{(n)})$, $n\geq 1$, be a sequence of super--replicating strategies
for which
$\lim _{n\rightarrow\infty}C(\pi^{(n)})= V^{\mathbb P}_{h_1,...,h_N}(H).$
Clearly $V^{\mathbb P}_{h_1,...,h_N}(H)\leq ||H||_{\infty}$. Hence without loss of generality,
we assume that for any $n$, $C(\pi^{(n)})<||H||_{\infty}+1$.
Let us prove that the sequence $c^{(n)}\in\mathbb R^{N+1}$, $n\in\mathbb N$, is bounded. Choose $n\in\mathbb N$.
We deduce from Assumption \ref{asm.noarbitrage}  that there exists
a probability measure $Q\in\mathcal M$ such that for any $i=1,\dots,N$,
$$ \mathbb E_{Q} [h_i(\mathbb S)]=
\begin{cases}
\mathcal P_i-\varepsilon & \text{if} \ c^{(n)}_i\geq 0   \\
\mathcal P_i+\varepsilon   & \text{if} \ c^{(n)}_i<0.
\end{cases}
$$
Lemma~\ref{lem.density} implies that there exists a probability
measure
$\mathbb Q\in \mathcal Q$ such that
$ \mathbb E_{\mathbb Q} [h_i(S)]<
\mathcal P_i-\varepsilon/2$ if $c^{(n)}_i\geq 0$ and
$ \mathbb E_{\mathbb Q} [h_i(S)]>
\mathcal P_i+\varepsilon/2$ if $c^{(n)}_i<0$. Thus, using the supermartingale property of each $\int \gamma^{(n)}\, dS$ under $\mathbb Q$, we obtain
\begin{align}
||H||_{\infty}+1&\geq C(\pi^{(n)}) \label{3+.6-}\\
&\geq c^{(n)}_0+
\mathbb E_{\mathbb Q} [\sum_{i=1}^N c^{(n)}_i h_i(S)]+\frac{\varepsilon}{2}\sum_{i=1}^N |c^{(n)}_i|\nonumber\\
&\geq
\mathbb E_{\mathbb Q} [H(S)-\int_{0}^T\gamma^{(n)}_t dS_t]+\frac{\varepsilon}{2}
\sum_{i=1}^N |c^{(n)}_i|\nonumber\\
&\geq -||H||_{\infty}+\frac{\varepsilon}{2}\sum_{i=1}^N |c^{(n)}_i|. \nonumber
\end{align}
From (\ref{3+.6-}), we derive that $|c^{(n)}_i|\leq \frac{2(1+2||H||_{\infty})}{\varepsilon}$ for all $n\in \mathbb N,i=1,\dots,N$. Moreover, by applying (\ref{3+.6-}) again we get that $c^{(n)}_0$ is uniformly bounded (in $n$).
  We conclude the uniform boundedness of $c^{(n)}$ as required.
Thus, there exists a subsequence (for simplicity we still denote it by $n$) such that
$\lim_{n\rightarrow\infty}c^{(n)}=\hat c=(\hat c_0,...,\hat c_N).$

Next, we apply the Koml\'{o}s theorem. Set
$Z_n=\int_{0}^T \gamma^{(n)}_t dS_t$, $n\in\mathbb N$.
Clearly $Z_n\geq H(S)-c_0-\sum_{i=1}^N c_i h_i(S)$ and so the sequence $Z_n$, $n\in\mathbb N$, is uniformly bounded from below.
Thus, by Delbaen and Schachermayer (1994, Lemma A 1.1) we obtain the existence of a sequence $\hat Z_n\in conv (Z_n,Z_{n+1},....)$, $n\in\mathbb N$,
such that $\hat Z_n$, $n\in\mathbb N$, converges a.s.
Denote the limit by $\hat Z$.
Using the fact that
the set of random variables which are dominated by stochastic integrals with respect to a
local martingale is Fatou closed, see Delbaen and Schachermayer (2006, Remark~9.4.3),
we can find a trading strategy
$\hat\gamma=\{\hat\gamma_{t}\}_{t=0}^T$ such that
$\int_{0}^t \hat\gamma_u dS_u$, $t\in [0,T]$ is uniformly bounded from below
and $\int_{0}^T\hat\gamma_t dS_t\geq \hat Z$.
Finally, we argue that $\hat\pi:=(\hat c,\hat\gamma)$ is an optimal super--replicating strategy.
Clearly, $C(\hat\pi)=\lim_{n\rightarrow\infty} C(\pi_n)=V^{\mathbb P}_{h_1,...,h_N}(H)$. Moreover,
it is straightforward to see that
$\int_{0}^T\hat\gamma_t dS_t\geq \hat Z\geq H(S)-\hat c_0-\sum_{i=1}^N \hat c_i h_i(S)$ a.s.,
and the result follows.
\end{proof}
\begin{rem}
A priori, it seems that we used a weaker assumption than Assumption \ref{asm.noarbitrage}.
Indeed, we only used that there exists $\varepsilon>0$ such that for any
$(j_1,...,j_N)\in\{-1,1\}^N$ there exists a probability measure $Q_{j_1,...,j_N}\in\mathcal M$
for which
$\mathbb E_{Q_{j_1,...,j_N}}[h_i(\mathbb S)]=\mathcal P_i+\varepsilon j_i$, $i=1,...,N$. However, by taking convex combinations
of such probability measures, we see that the weaker condition is in fact equivalent to Assumption \ref{asm.noarbitrage}.
\end{rem}
\begin{rem}
Let us remark that for the model free hedging, the existence of a
super--replicating strategy with minimal cost is an open question.
\end{rem}
\begin{rem}
Usually, the common static positions are call options.
However, due to the Put--Call parity, we can replace
the call options by put options and hence $h_1,...,h_N$ can be assumed to be bounded.
A natural question
is what if $H$ is unbounded, for instance if
$H(S)=\max_{0\leq t\leq T}S_t$ is a lookback option.
In this case we can show that if $h_1,...,h_N$ are bounded, then for fully incomplete markets
the super--replication is infinity. Namely, if the static positions are bounded, we cannot super--replicate a lookback option.
Thus, in order to have a reasonable super--replication price, we need to
assume that one of the $h_i$ is unbounded as well. For instance we can take a power option
$h_i(S)=S^p_T$, $p>1$. In this case Theorem
\ref{thm2.1} is much more delicate and in particular, requires some uniform integrability conditions.
Thus, the question whether Theorem
\ref{thm2.1} can be extended to the unbounded case remains open.
\end{rem}

\section{Hedging of Game Options}\label{sec:3++}\setcounter{equation}{0}
In this section, we deal with the super-replication of game options.
 Consider a financial market which is given by
 (\ref{2.1})--(\ref{2.2}). We assume that Definition \ref{dfn2.1} holds true, i.e., the
 market is fully incomplete.

Consider a game option with maturity date $T$ and payoffs which are given by
\begin{equation*}\label{def:game-option-payoff}
Y_t=f_1(S_t)\,\,\mbox{and}\,\, X_t=f_2(S_t),\, t\in [0, T],
\end{equation*}
where $f_1,f_2:\mathbb{R}_{+}\rightarrow\mathbb{R}_{+}$ are
continuous functions with $f_1\leq f_2$.
In addition\added{,} we assume that there exists $L>1$ such that for all $x,y>0$
\begin{equation}\label{2.3++}
|f_i(x)-f_i(y)|\leq L|x-y|\left(1+\frac{f_i(x)}{x}+\frac{f_i(y)}{y}\right), \ \ i=1,2.
\end{equation}
The condition (\ref{2.3++}) is weaker than assuming Lipschitz continuity,
and allows to consider Power options (in addition to \added{e.g.} call and put options).
We deduce from (\ref{2.3++}) that for any $x>0$
$$f_i\left(\frac{2L}{2L-1}x\right)\leq 2\left(\frac{L}{2L-1}x+\left(1+\frac{L}{2L-1}\right)f_i(x)\right), \ \ i=1,2.$$
For $\hat f_i(x):=\max(x,f_i(x))$, $i=1,2$ we obtain
$$\hat f_i\left(\frac{2L}{2L-1}x\right)\leq 2\left(\frac{L}{2L-1}+1+\frac{L}{2L-1}\right)\hat f_i(x)=\frac{8L-2}{2L-1}\hat f_i(x)$$
and so
$\hat f_i(x)\leq \max_{0\leq y\leq 1}\hat f_i(y)\left(\frac{8L-2}{2L-1}\right)^n$ for
 $\left(\frac{2 L}{2L-1}\right)^{n-1}\leq x\leq \left(\frac{2L}{2L-1}\right)^{n}$, $n\in\mathbb N$.
We conclude that there exists $\tilde L, N>1$ such that
for any $x>0$
\begin{equation}\label{2.4}
f_i(x)\leq \hat f_i(x)\leq \tilde L(1+x^N), \ \ i=1,2.
\end{equation}
Next, we introduce the notion of hedging. Recall
$\tilde  S_t=\frac{S_t}{B_t}$, $t\in [0,T]$, the discounted stock price, which
by (\ref{2.1})--(\ref{2.2})\added{,} has dynamics
$d\tilde S_t=\nu_t \tilde S_t \, dW_t$.
A self financing portfolio with an initial capital $z$ is a pair $\pi=(z,\gamma)$
where ${\{\gamma_t\}}_{t=0}^T$ is a progressively measurable process which satisfies
$\int_{0}^T\gamma^2_t \nu^2_t \tilde S^2_t\added{\,dt}<\infty$ a.s.
The corresponding portfolio value is given by
\begin{equation}\label{2.4+}
Z^\pi_t=B_t\left(\frac{z}{B_0}+\int_{0}^t \gamma_u \,d\tilde S_u\right)=
B_t\left(\frac{z}{B_0}+\int_{0}^t \gamma_u \tilde S_u \nu_u \,dW_u\right), \ \ t\in [0,T].
\end{equation}

As usual\deleted{,} for game options\added{,} a hedging strategy
consists of a self financing portfolio and a cancellation time. Thus, formally,
 a hedging strategy is a pair $(\pi,\sigma)$ such that
 $\pi$ is a self financing portfolio and $\sigma\leq T$ is a stopping time.
A hedging strategy $(\pi,\sigma)$ is super--replicating the game option if
for any $t\in [0,T]$
\begin{equation}\label{2.5}
Z^\pi_{t\added{\wedge\sigma}}\geq f_2(S_{\sigma})\mathbb{I}_{\sigma<t}+f_1(S_t)\mathbb{I}_{t\leq\sigma} \ \ \mbox{a.s.}
\end{equation}
The portfolio value process ${\{Z^\pi_t\}}_{t=0}^T$ is continuous and so,
if (\ref{2.5}) holds true for any $t\in [0,T]$, then
$$\mathbb P\big(\forall{t}\in [0,T], Z^\pi_{t\added{\wedge \sigma}}  \geq  f_2(S_{\sigma})\mathbb{I}_{\sigma<t}+f_1(S_t)\mathbb{I}_{t\leq\sigma}\big)=1.$$

 A hedging strategy $(\pi,\sigma)$ will be called trivial if it is of the form
\begin{equation}\label{2.6}
\gamma\equiv\gamma_0, \ \ \mbox{and} \ \ \sigma=\inf\{t: S_t\notin D\}\wedge{T}
\end{equation}
where $D\subset\mathbb R$ is an interval (not necessarily finite).

Define the super--replication price
\begin{equation*}
V=\inf\{Z^\pi_0:\exists\,\mbox{hedging strategy }\,(\pi,\sigma) \mbox{ super-replicating the option}
\}.
\end{equation*}
Also, set
\begin{equation*}
\mathbf V=\inf\{Z^\pi_0:\exists\,\mbox{trivial hedging strategy }\,(\pi,\sigma) \mbox{ super-replicating the option}\}.
\end{equation*}
Clearly the investor can cancel at $\sigma=0$ and so
 $V\leq \mathbf V\leq f_2(S_0)$.

Introduce the set $\mathbb H$ of all continuous functions
$h:(0,\infty)\rightarrow\mathbb R$ such that $f_1\leq h\leq f_2$ and
$h$ is concave in every interval in which $h<f_2$. \replaced{We deduce from}{From} Ekstr\"om and Villeneuve (2006, Lemma~2.4) that there exists a smallest element in $\mathbb H$ and \replaced{which is equal}{equals} to
$$g(x):=\inf_{h\in \mathbb H} h(x).$$

\replaced{Throughout this section, we}{We} will assume the following.
\begin{asm}\label{asm2.2}
At least one of the following conditions hold.\\
i. The interest rate is zero, i.e., $r\equiv 0.$\\
ii. For the initial stock price $S_0$ we assume that
if $g(S_0)<f_2(S_0)$, then
$$g(S_0)-S_0\partial _{+} g(S_0) \geq 0,$$
where $\partial _{+} g(S_0)$ is the right derivative at $S_0$\deleted{,} \added{(}\replaced{which}{it} exists because $g$ is concave in a
neighbourhood of $S_0$\added{)}.
\end{asm}
In Subsection \ref{sec:examples}, we
analyze in details the second condition in
Assumption \ref{asm2.2}.
In particular\added{,} we will see that it is satisfied for
most of the common payoff functions.

Next, for any $x\in \mathbb R_{+}$ introduce the open interval
$$K_x=\big(\sup\{z\leq x: g(z)=f_2(z)\},\,\inf\{z\geq  x: g(z)=f_2(z)\}\big)$$
where as usual, supremum and infimum\deleted{,} over an empty set are equal to
$-\infty$ and $\infty$, respectively.
Define the stopping time
\begin{equation*}
\hat\sigma=\inf\{t:S_t\not\in K_{S_{0}}\}\wedge T,
\end{equation*}
where we set $\hat\sigma=0$ if the set $K_{S_0}$
is empty \added{(where $(a,a):=\emptyset$ for any constant $a\in \mathbb{R}$)}.

The following theorem is the main result of this section.
It says that in fully incomplete markets, the super--replication
price of a game option is the cheapest cost of a trivial
super--replication hedging strategy, which can be calculated explicitly.
\begin{thm}\label{thm2.1game}
The super--replication price of the game
option introduced above is given by
\begin{equation*}
\mathbf V= V=g(S_0).
\end{equation*}
Furthermore,
define the buy--and--hold portfolio
strategy $\hat\pi=(g(S_0),\hat\gamma)$ by
\begin{equation*}
\hat\gamma\equiv\left\{ \begin{array}{ll}
 \partial_{+}g(S_0)\ \ & \mbox{if } \  g(S_0)<f_2(S_0),\\
 0\ \ & \mbox{otherwise.}
\end{array}\right.
\end{equation*}
Then $(\hat\pi,\hat\sigma)$ is the cheapest hedging strategy super-replicating the option.
\end{thm}
\begin{proof}
As $\mathbf V\geq V$, Theorem \ref{thm2.1game} will follow from the inequality
\begin{equation}\label{2.30}
V\geq g(S_0)
\end{equation}
and the fact that $(\hat\pi,\hat\sigma)$
is a super-replicating strategy.
\replaced{Inequality}{The inequality} (\ref{2.30})
is the difficult part and will be proved in
Section \ref{sec:proof}. The fact that $(\hat\pi,\hat\sigma)$ is a super-replicating strategy
is simpler and we provide its proof here.

First, if $g(S_0)=f_2(S_0)$, then the statement is trivial. \replaced{Therefore}{Thus},
assume that
$g(S_0)<f_2(S_0)$.
Let $t\in [0,T]$. Observe that on the event \replaced{$\hat\sigma<t$}{$t<\hat\sigma$},
$g(S_{\hat\sigma})=f_2(S_{\hat\sigma})$.
From
Assumption \ref{asm2.2}, it follows that if
$\frac{B_{t\wedge\hat\sigma}}{B_0}>1$ then $g(S_0)-S_0\partial_{+} g(S_0)\geq 0$.
This together with
the fact that $g$ is concave in the interval $K_{S_0}$ yields
\begin{align}\label{mew}
Z^{\hat\pi}_{t\wedge\hat\sigma}&=\frac{B_{t\wedge\hat\sigma}}{B_0}\big(g(S_0)-S_0\partial_{+}
g(S_0)\big)+\partial_{+} g(S_0)S_{t\wedge\hat\sigma}\\
&\geq g(S_0)+\partial_{+} g(S_0)(S_{t\wedge\hat\sigma}-S_0)
  \geq g(S_{t\wedge\hat\sigma})\geq
f_2(S_{\hat\sigma})\mathbb{I}_{\hat\sigma<t}+f_1(S_t)
\mathbb{I}_{t\leq\hat\sigma}\replaced{.}{,}\nonumber
\end{align}
\end{proof}
\begin{rem}\label{rem4.3}
Let us notice that (\ref{mew}) holds true pathwise, and hence
the hedging strategy $(\hat\pi,\hat\sigma)$
is a super--replicating one in the model free sense.
Thus, from Theorem \ref{thm2.1game} we conclude that for
fully incomplete markets the super--replication price coincides
with the model free super--replication price.
\end{rem}
\begin{rem}
In Example \ref{ex:call} we will see that without the second part of Assumption \ref{asm2.2},
the hedge $(\hat\pi,\hat\sigma)$ may not be super--replicating, and so
Theorem \ref{thm2.1game} may not hold true.
\end{rem}

\vspace{10pt}
\subsection{Examples}\label{sec:examples}
In this subsection, we give several examples for applications of Theorem~\ref{thm2.1game}.
In the case where both $f_1$ and $f_2$ are convex, we can calculate $g(S_0)$ and $\partial_+ g(S_0)$ explicitly. To this end, we assume throughout this
subsection that $f_1$ and $f_2$ are convex functions. Set
\begin{equation}\label{eq:Def-A}
A=\left\{ \begin{array}{ll}
\inf\big\{y>0 : \frac{f_2(y)-f_1(0)}{y}\leq \partial_{+}f_2(y)\big\} &\mbox{if } \ f_1(0)<f_2(0)\\
0 &\mbox{if } \ f_1(0)=f_2(0),
\end{array}\right.
\end{equation}
as well as
\begin{equation*}
\beta=\left\{\begin{array}{ll} \frac{f_2(A)-f_1(0)}{A} \, \mathbb{I}_{A<\infty} + \infty \, \mathbb{I}_{A=\infty} & \ \mbox{if } \ f_1(0)<f_2(0)\\
\partial_{+}f_2(0) & \ \mbox{if } \ f_1(0)=f_2(0).
\end{array}\right.
\end{equation*}
Moreover, set
\begin{equation*}
m:=\lim_{t\rightarrow\infty} \partial_{+}f_1(t), \quad \quad \rho:=\inf\{t :  \partial_{+}f_2(t)>m \}.
\end{equation*}
Observe that the terms $A,\beta,m,\rho$ can take the value $\infty$. Moreover,
if $m=\infty$, then $\lim_{t\rightarrow\infty} \partial_{+}f_2(t)=\infty$ as well.
In this case, from the convexity of $f_2$
$$\lim_{t\rightarrow\infty} f_2(t)-t\partial_{+}f_2(t)\leq
\lim_{t\rightarrow\infty}f_2(1)+(t-1)\partial_{+}f_2(t)-t\partial_{+}f_2(t)=
-\infty.$$
Thus $A,\beta<\infty$. We conclude that in any case $\beta\wedge m<\infty$.

Define the function $g:\mathbb{R}_{+}\rightarrow\mathbb{R}_{+}$
by
\begin{equation}\label{eq:g-convex-def}
g(x)=\left\{ \begin{array}{ll}
 (f_1(0) + \beta x)\mathbb{I}_{x<A}+
 f_2(x)\mathbb{I}_{A\leq x<\rho}+(f_2(\rho)+m(x-\rho))\mathbb{I}_{x\geq\rho}  &\mbox{if } \ \beta<m\\
f_1(0) + m x &\mbox{if } \ m\leq \beta. \ \ \phantom{m<\infty; \,}
        \end{array}\right.
\end{equation}
\begin{lem}\label{le:min-g}
If both $f_1$ and $f_2$ are convex, then the function $g$ defined in \eqref{eq:g-convex-def} is the minimal element in $\mathbb{H}$.
 \end{lem}
\begin{proof}
 By definition, we see that $g \in \mathbb{H}$. Denote by $g_{min}$ the minimal element of $\mathbb{H}$.
 Then, $g_{min}(0)=f_1(0)=g(0)$.
 Assume by contradiction that there exists $x>0$ for which $g_{min}(x)<g(x)$.
Set,
 $$y=\inf\{t<x: g_{min}(t)<g(t) \  \mbox{on}  \ \mbox{the}  \ \mbox{interval} \ (t,x)\}$$
 and $$z=\sup\{t>x:g_{min}(t)<g(t) \  \mbox{on}  \ \mbox{the}  \ \mbox{interval} \ (x,t)\}.$$
  By continuity of $g_{min},g,$
  we have $y<x<z$. By definition of $\mathbb{H}$, $g_{min}$ is concave on $I:=(y,z)$ as $g_{min}<g\leq f_2$ on $I$.
  Observe that $g$ is convex on $\mathbb{R}_+$. Therefore, if $z<\infty$ we would get that
  $g-g_{min}$ is a convex function which is strictly positive on $I$ and satisfies
  $g(z)-g_{min}(z)=g(y)-g_{min}(y)=0$. But this is not possible and we conclude that $z=\infty$. Thus, $g_{min}<f_2$ on $I=(y,\infty)$ and so $g_{min}$ is concave
  on $(y,\infty)$. This together with the fact that $g_{min}\geq f_1$ gives
  $\inf_{t>y}\partial_{+}g_{min}(t)\geq m$.
  We derive from (\ref{eq:g-convex-def}) that
   $\sup_{t>0}\partial_{+}g(t)\leq m$.
  Thus,
  $g_{min}-g$ is non decreasing in the interval $(y,\infty)$, and so from
  the equality $g(y)-g_{min}(y)=0$ we conclude that $g_{min} \geq  g$ on $(y,\infty)$, this is
   a contradiction.
  \end{proof}
We obtain from (\ref{eq:g-convex-def}) that if the initial stock price satisfies
$S_0\leq \rho$, then the second condition in Assumption \ref{asm2.2} is satisfied.
In particular, if $\rho=\infty$ this holds true trivially. Observe that
$$
\rho=\infty \Leftrightarrow \sup_{t>0}\partial_{+}f_1(t)=\sup_{t>0}\partial_{+}f_2(t).$$
This brings us to the following
immediate\deleted{(without proof)} Corollary.
 \begin{cor}\label{co:constant-penalty}
 ${}$\\
If at least one of the below conditions holds: \\
i. $f_2(x)=f_1(x)+\Delta$ for some constant $\Delta>0$ (i.e. constant penalty),\\
ii. $\sup_{t>0}\partial_{+}f_2(t)=0$ (for instance Put options),\\
iii. $\sup_{t>0}\partial_{+}f_1(t)=\infty$ (for instance Power options),\\
then the second condition in Assumption \ref{asm2.2} is satisfied.
 \end{cor}
Next, we give several explicit
examples
for applications of Theorem~\ref{thm2.1game}.
Given $f_1(x)$ convex, let $f_2(x)= c f_1(x) + \Delta$, where $c\geq 1, \Delta \geq0$. 
Recall the game trading strategy $(\hat{\pi},\hat{\sigma})$ which was defined in Theorem~\ref{thm2.1game}.
%
\begin{example}[Call option]\label{ex:call}
Let $K>0$ be a constant. Consider a game call option
\begin{equation*}
f_1(S_t)=(S_t-K)^+,  \quad f_2(S_t)=c (S_t-K)^+ + \Delta.
\end{equation*}
%
We distinguish between two cases.
\begin{enumerate}
\item $\Delta< K$:\,
In this case,
\begin{equation*}
A=\left\{\begin{array}{ll}
 K  &\mbox{if } \ \Delta>0 \\
0 &\mbox{if } \ \Delta=0,
        \end{array}\right.
\end{equation*}
$\beta=\frac{\Delta}{K}<m=1$ and
\begin{equation*}
\rho=\left\{ \begin{array}{ll}
 \infty  &\mbox{if } \ c=1 \\
K &\mbox{if } \ c>1.
        \end{array}\right.
\end{equation*}
Thus, see Figure~\ref{pic:Callneu1} and Figure~\ref{pic:Callneu3}, we have
\begin{equation*}
g(S_0)=\frac{\Delta}{K} S_0 \,\mathbb{I}_{S_0< K} +\big( S_0 -K + \Delta\big)\,\mathbb{I}_{S_0\geq K}.
\end{equation*}
Moreover,
\begin{enumerate}
\item If $S_0\leq K$, then
\begin{equation*}
(\hat\pi,\hat\sigma)=\left\{ \begin{array}{ll}
\big((\frac{\Delta}{K} S_0,\frac{\Delta}{K}),\,\inf\{t : S_t= K \}\wedge T\big) &\mbox{if } \ \Delta>0\\
\big((0,0),\,0\big) &\mbox{if } \ \Delta=0.
\end{array}\right.
\end{equation*}
%
\item If $S_0> K$, then
\begin{equation*}
(\hat\pi,\hat\sigma)=\left\{ \begin{array}{ll}
\Big( \big(S_0 -K + \Delta,0\big),0\Big)  &\mbox{if } \ c=1 \\
\Big( \big(S_0 -K + \Delta,1\big),\inf\{t : S_t= K \}\wedge T\Big) &\mbox{if } \ c>1.
        \end{array}\right.
\end{equation*}
Observe that for the case $c>1$ and $S_0>K$, the second condition in
 Assumption \ref{asm2.2} is not satisfied. Thus, in order for Theorem \ref{thm2.1game} to hold true, we need to take the
 interest rate
 $r\equiv 0$. Indeed, for $r>0$ we get that the portfolio value of $\hat\pi$ equals $Z^{\hat\pi}_t=S_t-\frac{B_t}{B_0}(K-\Delta)$. It follows that if
 $\frac{B_t}{B_0}(K-\Delta)>K$  then
 $Z^{\hat\pi}_t < S_t-K$, and so $(\hat\pi,\hat\sigma)$ is not a super-replicating strategy.
\end{enumerate}
%
%
%
%
\item $\Delta\geq K$:\,
In this case $A=K$, $\beta=\frac{\Delta}{K}\geq m=1$.
Thus, see Figure~\ref{pic:Callneu2}, we have $g(S_0)=S_0$.
Moreover, 
 $(\hat{\pi},\hat{\sigma})=\big( (S_0,1),T\big)$.
\end{enumerate}
\end{example}
%

\begin{figure}[!ht]
    \subfloat[If\, $\Delta>0,\,\Delta< K$\label{pic:Callneu1}]{%
      \includegraphics[width=0.326\textwidth]{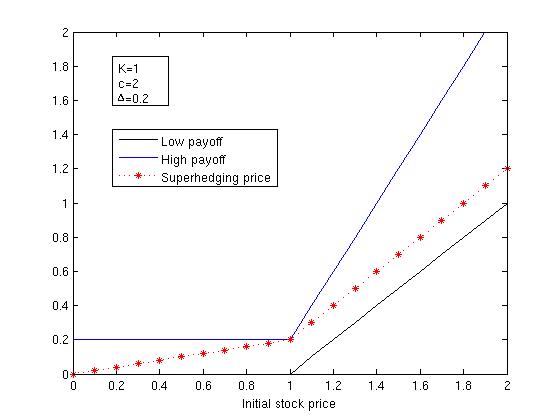}
    }
    \hfill
        \subfloat[If \,$\Delta=0,\,c>1$\label{pic:Callneu3}]{%
          \includegraphics[width=0.326\textwidth] {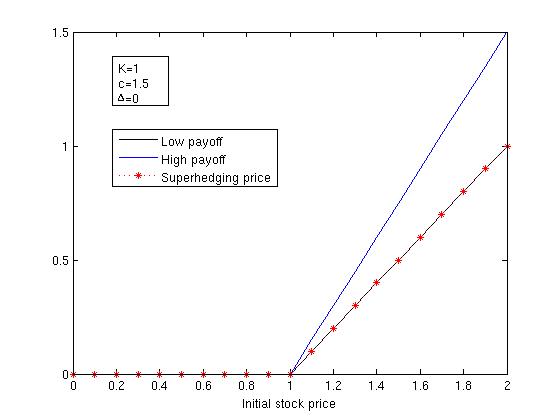}
          }
           \hfill
              \subfloat[If \, $\Delta\geq K$\label{pic:Callneu2}]{%
                \includegraphics[width=0.326\textwidth] {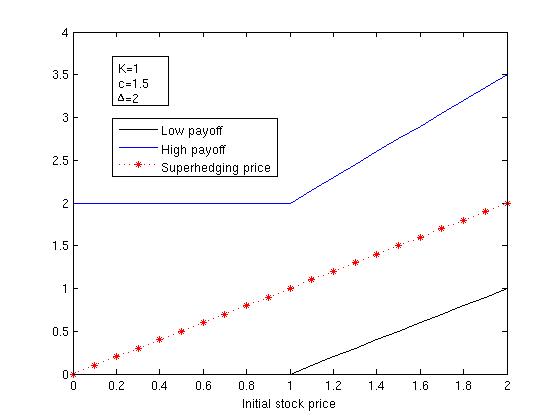}
              }
    \caption{Call option}
  \end{figure}
\begin{example}[Put option]\label{ex:Put}
Let $K>0$ be a constant. Consider a game put option
\begin{equation*}
f_1(S_t)=(K-S_t)^+,  \quad f_2(S_t)=c(K-S_t)^+ + \Delta.
\end{equation*}
 We distinguish between two cases.
%


\begin{enumerate}
\item $\Delta<K$:\,
In this case
$A=K$, $\beta=\frac{\Delta-K}{K}<m=0$ and $\rho=\infty$. Hence,
see Figure~\ref{pic:Putneu1} and Figure~\ref{pic:Putneu3}, the super-replication price is
\begin{equation*}
g(S_0)=\Big(K-\frac{K-\Delta}{K} S_0\Big)\,\mathbb{I}_{S_0<K} + \Delta \,\mathbb{I}_{S_0\geq K}.
\end{equation*}
\begin{enumerate}
\item If $S_0<K$, then 
$(\hat{\pi},\hat{\sigma})=\big((K-\frac{K-\Delta}{K} S_0,-\frac{K-\Delta}{K}),\,\inf\{ t : S_t=K\}\wedge T \big).$
%
\item If $S_0\geq K$, then
$(\hat{\pi},\hat{\sigma})=\big((\Delta,0),0\big)$.
\end{enumerate}
%
%
\item $\Delta\geq K$\,
In this case,
\begin{equation*}
A=\left\{\begin{array}{ll}
 K  &\mbox{if } \ \Delta=K \\
\infty &\mbox{if } \ \Delta>K,
        \end{array}\right.
\end{equation*}
and
\begin{equation*}
\beta=\left\{\begin{array}{ll}
 0  &\mbox{if } \ \Delta=K \\
\infty &\mbox{if } \ \Delta>K.
        \end{array}\right.
\end{equation*}
Thus $\beta\geq m=0$. Hence,
see Figure~\ref{pic:Putneu2}, the super-replication price equals $g(S_0)\equiv K$,
and  $(\hat{\pi},\hat{\sigma})=\big((K,0),T\big)$.
\end{enumerate}
In  the Put-case, the super-replication price is independent of the scaling factor $c\geq1$.
\end{example}
%
%
\begin{figure}[!ht]
    \subfloat[If \,$\Delta>0,\,\Delta<K$\label{pic:Putneu1}]{%
      \includegraphics[width=0.33\textwidth]{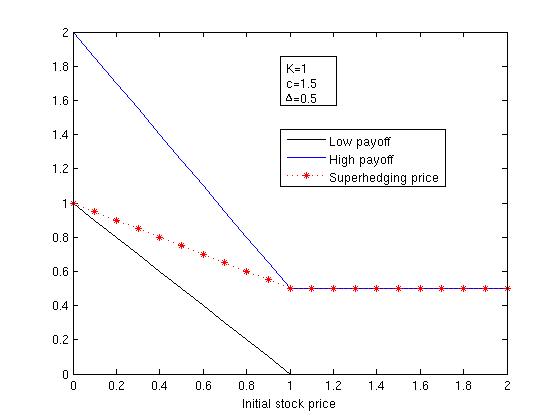}
    }
    \hfill
        \subfloat[If \,$\Delta=0,\,c>1$\label{pic:Putneu3}]{%
          \includegraphics[width=0.33\textwidth] {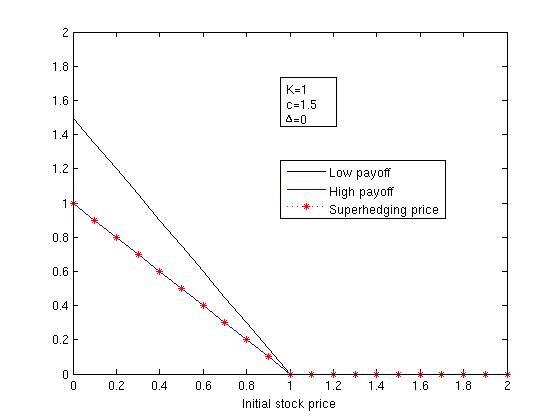}
        }
         \hfill
            \subfloat[If \,$\Delta\geq K$\label{pic:Putneu2}]{%
              \includegraphics[width=0.33\textwidth] {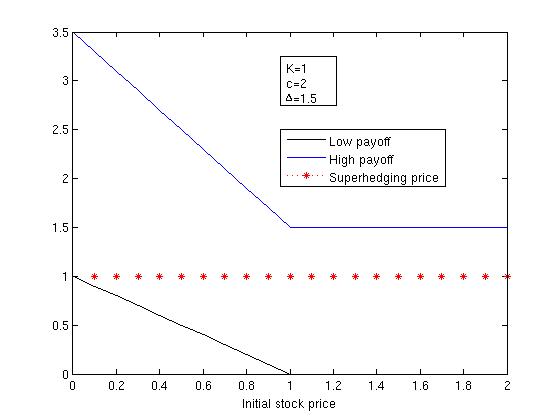}
            }
    \caption{Put option}
  \end{figure}
%
\begin{example}[Power option]
%
%
Let $p>1$ and consider the game $p$-th power option
\begin{equation*}
f_1(S_t)=S_t^p,  \quad f_2(S_t)= cS_t^p + \Delta.
\end{equation*}
We have $\rho=m=\infty$ and when $\Delta>0$, $A=\left(\frac{\Delta}{c(p-1)}\right)^{1/p}$, $\beta= c p \left(\frac{\Delta}{c(p-1)}\right)^{1-1/p}$. Thus,  see Figure~\ref{pic:Powerneu1} and Figure~\ref{pic:Powerneu2}, the super--replication price equals
\begin{equation*}
g(S_0)= \beta S_0 \,\mathbb{I}_{S_0<A}+ (c S_0^p+\Delta)\,\mathbb{I}_{S_0\geq A}.
\end{equation*}
The cheapest super-replicating strategy is given by:
\item If $S_0<A$, then  $(\hat{\pi},\hat{\sigma})= \big((\beta S_0,\beta),\, \inf\{t : S_t = A\}\wedge T\big)$.
\item If $S_0\geq A$, then  $(\hat{\pi},\hat{\sigma})=\big((cS_0^p+\Delta,0),0\big)$.

%

\end{example}
 \begin{figure}[!ht]
 \centering
     \subfloat[If $\Delta>0$\label{pic:Powerneu1}]{%
       \includegraphics[width=0.327\textwidth]{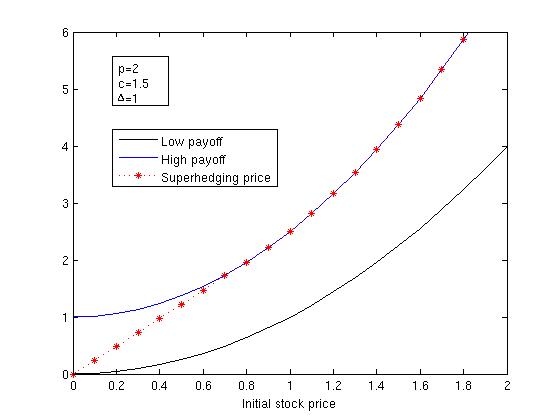}
     }
     \hfill
     \subfloat[If $\Delta=0, \, c>1$\label{pic:Powerneu2}]{%
           \includegraphics[width=0.327\textwidth]{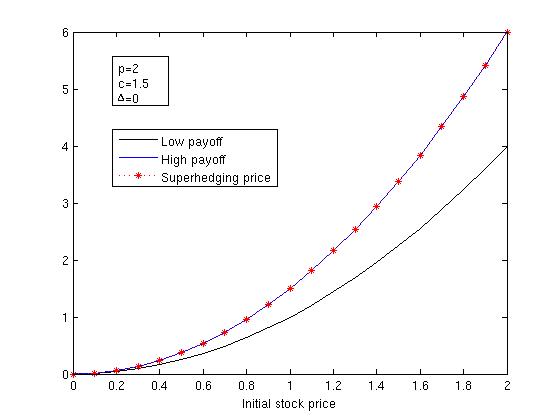}
         }
     \caption{Power option}
   \end{figure}
%

\section{Proof of Proposition \ref{prop2.1}}\label{sec:3}\setcounter{equation}{0}
\begin{proof}
Let $\alpha\in\mathcal C(\nu_0)$ and $\epsilon>0$.
We will show (for both set-ups I and II) that there exists a probability measure
$\mathbb Q\ll \mathbb P$ such that the properties of Definition \ref{dfn2.1} hold true.\\
I.
There \replaced{is}{exists} a constant $C>0$ such that
$\frac{1}{C}\leq\alpha\leq C$. Without loss of generality
we assume that
$\epsilon<\frac{1}{2C}.$
Define the stopping time
\begin{equation*}
\Theta=\inf\{t: |\alpha_t-\nu_t|\replaced{\geq}{=}\epsilon\}\wedge T.
\end{equation*}
Clearly,
$$0<\frac{1}{2 C}\leq\inf_{0\leq t\leq \replaced{\Theta}{\theta}}\nu_{t}\leq \sup_{0\leq t\leq \replaced{\Theta}{\theta}}\nu_t\leq C+\frac{1}{\added{2}C}.$$
From the assumptions on the functions $a,b,c$,
we get that there exists a constant $\tilde C>0$ such that
\begin{equation}\label{3.1}
\sup_{0\leq t\leq \Theta}\left[|a(t,\nu_t)|+|b(t,\nu_t)|+|c(t,\nu_t)|+\frac{1}{|b(t,\nu_t)|}\right]\leq \tilde C.
\end{equation}
Fix $n>\frac{2\tilde C}{T}$.
For $k=1,...,n$\added{,} let
\begin{align*}
I_k&=\int_{(k-1)T/n}^{kT/n} a(t,\nu_t)\,dt+\int_{(k-1)T/n}^{kT/n} c(t,\nu_t)\, dW_t, \\
 J_k&=\alpha_{\frac{kT}{n}}-\alpha_{\frac{(k-1)T}{n}}.\nonumber
\end{align*}
Introduce the function
$\Phi(x)=-n^2\vee (x\wedge n^2)$, $x\in\mathbb R$.
Let ${\{\gamma_t\}}_{t=0}^T$ and ${\{\tilde W_t\}}_{t=0}^T$ be the unique stochastic processes which satisfy
the following (recursive) relations:
$$\hat W_t=\tilde W_t+\int_{0}^t\gamma_u \,du$$
where $\gamma_t=0$ for $t\leq \frac{T}{n}$, and for $k=1,...,n\added{-1}$
$$\gamma_t=\Phi\left(\frac{n}{b(t,\nu_t)T}
\Big(J_k-I_k-\int_{(k-1)T/n}^{kT/n} b(u,\nu_u)\,d\tilde W_u\Big)\right), \  \frac{kT}{n}<t\leq \frac{(k+1)T}{n}.$$

The process $\{\gamma_t\}_{t=0}^T$ is uniformly bounded, thus we deduce from the Girsanov theorem \added{and Novikov condition} that there exists a probability measure $\mathbb Q\sim\mathbb P$ (which depends on $n$)
such that
$\{(\tilde W_t,W_t)\}_{t=0}^T$ is a two dimensional
standard Brownian motion with respect to $\mathbb Q$ and the filtration
${\{\mathcal F_t\}}_{t=0}^T$.

For any $k=1,...,n$\added{,} denote
$L_k=\int_{(k-1)T/n}^{kT/n} b(t,\nu_t)\,d\tilde W_t$
and introduce the event
$$A_k=\{kT/n<\Theta\}\replaced{\cap}{\bigcap} \left\{|I_k|+|L_k|>1\right\}.$$
Clearly,
for any $k=1,...,n$
\begin{equation*}
\mathbb{I}_{\frac{kT}{n}<\Theta} \,|I_k|\leq
\Big|\int_{(k-1)T/n}^{kT/n} \mathbb{I}_{t<\Theta} \,a(t,\nu_t)\,dt\Big|+
\Big|\int_{(k-1)T/n}^{kT/n} \mathbb{I}_{t<\Theta}\,c(t,\nu_t)\,dW_t\Big|.
\end{equation*}
This together with (\ref{3.1}) and the Burkholder--\replaced{Davis}{David}--Gundy inequality yield\deleted{s}
that for any $p>1$ there exists a constant $c_p>0$
such that
\begin{equation}\label{3.5}
\mathbb E_{\mathbb Q}\left[\mathbb{I}_{\frac{kT}{n}<\Theta}\, |I_k|^p\right]\leq 2^{p}\Big((\replaced{\tilde{C}}{\tilde N} T/n)^p +c_p (\replaced{\tilde{C}^2}{\tilde N^2} T/n)^{p/2}\Big).
\end{equation}
Similarly,
\begin{equation}\label{3.6}
\mathbb E_{\mathbb Q}\left[\mathbb{I}_{\frac{kT}{n}<\Theta} |L_k|^p\right]\leq
\mathbb E_{\mathbb Q}\left[\Big|\int_{(k-1)T/n}^{kT/n} \mathbb{I}_{t<\Theta}b(t,\nu_t)\,d\tilde W_t\Big|^p\right]\leq
c_p (\replaced{\tilde{C}^2}{\tilde N^2} T/n)^{p/2}.
\end{equation}
By applying the Markov inequality and (\ref{3.5})--(\ref{3.6}) for $p=4$, we obtain
\begin{equation}\label{3.7}
\mathbb Q\big(\cup_{k=1}^n A_k\big)\leq \sum_{k=1}^n \mathbb Q(A_k)\leq \frac{c}{n}
\end{equation}
for some constant $c$ \added{(independent of $n$)}.

Next, let $k<n$
and $kT/n\leq t<(k+1)T/n$.
Consider the event
$$U:=\{t<\Theta\}\setminus \Big(\big(\cup_{j=1}^n A_j\big)\cup \big(\max_{|u-v|\leq\frac{T}{n}} |\alpha_u-\alpha_v|>1\big)\Big).$$
Recall the constant $\tilde C$
from (\ref{3.1}). As $n>\frac{2\tilde C}{T}$, we get on the event $U$ that
for any $u\leq t$
 $$\gamma_u b(u,\nu_u)=\frac{n(J_m-I_m-L_m)}{T}, \quad \mbox{for } \ \frac{mT}{n}< u\leq \frac{(m+1)T}{n},$$
 where we set $I_{0}=J_{0}=L_{0}=0$.
Thus, on the event $U$ we have
\begin{equation*}
\nu_{\frac{kT}{n}}-\alpha_{\frac{kT}{n}}=\sum_{m=1}^{k}[I_m+L_m-J_m]+\sum_{m=1}^{k-1}[J_m-I_m-L_m]=\mathbb{I}_{\frac{kT}{n}<\Theta}(I_k+L_k-J_k)
\end{equation*}
\replaced{as well as}{and}
\begin{align*}
|\nu_{t}-\nu_{\frac{kT}{n}}|&\leq \Big|\int_{kT/n}^t \mathbb{I}_{u<\Theta}\, a(u,\nu_u)\,du\Big|+
\Big|\int_{kT/n}^t \mathbb{I}_{u<\Theta}\, b(u,\nu_u)\,d\tilde W_u\Big|\\
& \quad +
\Big|\int_{kT/n}^t \mathbb{I}_{u<\Theta}\, c(u,\nu_u)\,dW_u\Big|+\mathbb{I}_{\frac{kT}{n}<\Theta}\,(|J_k|+|I_k|+|L_k|).
\end{align*}
We conclude that on the event
$\hat U:=\Omega\setminus \Big(\big(\cup_{j=1}^n A_j\big)\cup \big(\max_{|u-v|\leq\frac{T}{n}} |\alpha_u-\alpha_v|>1\big)\Big)$
\begin{align}\label{3.8}
& \sup_{0\leq t<\Theta}|\alpha_t-\nu_t| \\
\leq &\,\max_{|u-v|\leq\frac{T}{n}} |\alpha_u-\alpha_v|
+
2\max_{1\leq k\leq n}
\Big(\mathbb{I}_{\frac{kT}{n}<\Theta}\,(|J_k|+|I_k|+|L_k|)\Big)
+\max_{1\leq k\leq n} (\Gamma_k+\Upsilon_k+\Lambda_k) \nonumber\\
\leq &\,3\max_{|u-v|\leq\frac{T}{n}} |\alpha_u-\alpha_v|+2\max_{1\leq k\leq n}
\Big(\mathbb{I}_{\frac{kT}{n}<\Theta}\,(|I_k|+|L_k|)\Big)
+
\max_{1\leq k\leq n} (\Gamma_k+\Upsilon_k+\Lambda_k), \nonumber 
\end{align}
where
\begin{align*}
\Gamma_k&=\max_{(k-1)T/n\leq t\leq kT/n} \Big|\int_{(k-1)T/n}^t \mathbb{I}_{u<\Theta}\, a(u,\nu_u)\,du\Big|,\\
\Upsilon_k&=\max_{(k-1)T/n\leq t\leq kT/n} \Big|\int_{(k-1)T/n}^t \mathbb{I}_{u<\Theta}\, b(u,\nu_u)\,d\tilde W_u\Big|, \\
\Lambda_k&=\max_{(k-1)T/n\leq t\leq kT/n} \Big|\int_{(k-1)T/n}^t \mathbb{I}_{u<\Theta}\, c(u,\nu_u)\,dW_u\Big|.
\end{align*}
Similarly
to (\ref{3.5})--(\ref{3.6}), we get that
$$\mathbb{E}_{\mathbb Q}\left[\max_{1\leq k\leq n}(\deleted{|}\Gamma_k\deleted{|}+\deleted{|}\Upsilon_k\deleted{|}+\deleted{|}\Lambda_k\deleted{|})^4\right]\leq
3^4\sum_{k=1}^n \mathbb E_{\mathbb Q}\big[\Gamma^4_k+\Upsilon^4_k+\Lambda^4_k\big]\leq \frac{\tilde c}{n}$$
for some constant $\tilde c$. Thus\added{,} from the Markov inequality we get that for sufficiently large $n$
\begin{equation}\label{3.8+}
\mathbb{Q}\left(
\max_{1\leq k\leq n} [\Gamma_k+\Upsilon_k+\Lambda_k]\geq\frac{\epsilon}{3}\right)<\frac{\epsilon}{5}.
\end{equation}
Similarly,
(\ref{3.5})--(\ref{3.6}) give that for sufficiently large $n$
\begin{equation}\label{3.9}
\mathbb{Q}\left(2\max_{1\leq k\leq n}
\left[\mathbb{I}_{\frac{kT}{n}<\Theta}(|I_k|+|L_k|)\right]
\geq\frac{\epsilon}{3}\right)<\frac{\epsilon}{5}.
\end{equation}
The stochastic \added{process} $\alpha$ is progressively measurable with respect to
the filtration generated by $W$, thus the distribution\deleted{s} of $\alpha$ under $\mathbb Q$ \replaced{is the same as under}{and} $\mathbb P$ \deleted{are the same}
and so,
for sufficiently large $n$
\begin{equation}\label{3.10}
\mathbb Q\Big(3\max_{|u-v|\leq\frac{T}{n}} |\alpha_u-\alpha_v|\geq \frac{\epsilon}{3}\Big)
<\frac{\epsilon}{5}.
\end{equation}
Finally,
by combining (\ref{3.7})--(\ref{3.10})\added{,} we obtain that for sufficiently large $n$,
\begin{align*}
& \ \mathbb Q(\|\alpha-\nu\|_{\infty}>\epsilon)\\
\leq & \
\mathbb{Q}\Big(\big(\cup_{j=1}^n A_j\big)\cup \big(\max_{|u-v|\leq\frac{T}{n}} |\alpha_u-\alpha_v|>1\big)\Big)+
\mathbb{Q}\Big(\big(\sup_{0\leq t<\Theta}|\alpha_t-\nu_t|=\epsilon\big)\cap\hat U\Big)\\
\leq & \ \frac{c}{n}+\frac{\epsilon}{5}+
\frac{3\epsilon}{5}<\epsilon,
\end{align*}
as required.\qed
${}$\\
II.
Consider the continuous stochastic process
$\phi_t=\ln\alpha_t-\ln\nu^{(1)}_t$, $t\in [0,T]$. Fix $\delta>0$.
Choose $n\in\mathbb N$ sufficiently large such that
\begin{equation}\label{3.10+}
\mathbb P\Big(\max_{|u-v|\leq\frac{T}{n}} |\phi_u-\phi_v|
\geq \delta\Big)\leq \delta.
\end{equation}
For $k=0,...,n-1$,
define the events
\begin{equation*}
\hat A_k=\left\{\max_{k T/n\leq t\leq (k+1)T/n}
|\ln\nu^{(2)}_t-\ln\nu^{(2)}_{\frac{kT}{n}}-(nt/T-k)(\phi_{\frac{kT}{n}}-\phi_{\frac{(k-1)T}{n}})|
<\frac{\delta}{n}
\right\}
\end{equation*}
where we set $\phi_{-\frac{T}{n}}\equiv \phi_0$. First, we argue that for any $k$
\begin{equation}\label{3.11}
\mathbb P(\hat A_k\,|\,\mathcal F_{\frac{k T}{n}})>0 \  \mbox{a.s.}
\end{equation}
Denote by $\{\mathbf G_t\}_{t=0}^T$ the usual augmentation of the filtration generated
by $\nu^{(2)}$.
By
our assumptions,
$\nu^{(2)}$ is independent of $W$ and $\phi$.
This, together with the fact that
$\{\mathcal F_t\}_{t=0}^T$ is the usual augmentation of the filtration generated
by $W$ and $\nu^{(2)}$  yields
$$\mathbb P\left(\hat A_k|\mathcal F_{\frac{k T}{n}}\right)=\Psi\left(\phi_{\frac{(k-1)T}{n}},\phi_{\frac{k T}{n}},\nu^{(2)}\right) \ \mbox{ a.s.,}$$
where
$\Psi:\mathbb R\times \mathbb R\times C_+[0,T]\rightarrow\mathbb R$ is a measurable function satisfying a.s.
$$\Psi(u,v,\nu^{(2)})=
\mathbb P\left(\max_{k T/n\leq t\leq (k+1)T/n}|\ln\nu^{(2)}_t-\ln\nu^{(2)}_{\frac{kT}{n}}-\right.\\
\left.(nt/T-k)
(v- u)|<\frac{\delta}{n}\,\big|\,\mathbf G_{\frac{k T}{n}}\right).
$$
It is assumed $\ln \nu^{(2)}$ satisfies the CFS property with respect to its natural filtration.
We deduce from  Pakkanen (2010, Lemma 2.3)
that $\ln \nu^{(2)}$ satisfies the CFS property with respect to the usual augmented filtration $\{\mathbf G_t\}_{t=0}^T$, as well. Therefore, we obtain  that $\Psi(u,v,\nu^{(2)})>0 \  \mathbb P$-a.s., for any $u,v\in\mathbb R$, hence we conclude that (\ref{3.11}) holds true.

Next, define the continuous martingale $\mathbf Z=\{\mathbf Z_t\}_{t=0}^T$ by
$\mathbf Z_0=1$ and
\begin{equation*}
\mathbf Z_t=\frac{\mathbb P\big(\hat A_k\,|\,\mathcal F_{t}\big)}{\mathbb P\big(\hat A_k\,|\,\mathcal F_{\frac{k T}{n}}\big)}
\sum_{i=0}^{k-1} \frac{\mathbb I_{\hat A_i}}{\mathbb P\big(\hat A_i\,|\,\mathcal F_{\frac{i T}{n}}\big)},
\ \ t\in (kT/n, (k+1)T/n],\ \ 0\leq k\leq n-1.
 \end{equation*}
There exists a probability measure $\mathbb Q\ll \mathbb P$ such that
$\frac{d\mathbb Q}{d\mathbb P}{|}_{\mathcal F_t}=\mathbf Z_t$, $t\in [0,T]$.
Let us prove that (for sufficiently small $\delta>0$),
$\mathbb Q$ satisfies the required properties.
Fix $k<n$ and
$t\in [kT/n, (k+1)T/n]$. On the event $\mathbf Z_t\neq 0$, using that $W_{\frac{(k+1)T}{n}}-W_{t}$ is independent of
$\mathcal F_t$ and $\hat A_k$, yields
\begin{align*}
\mathbb E_{\mathbb Q}\left(W_{\frac{(k+1)T}{n}}-W_{t}\,|\,\mathcal{F}_t\right)&=\frac{1}{\mathbf Z_t}
\mathbb E_{\mathbb P}\left(\mathbf Z_{\frac{(k+1)T}{n}}(W_{\frac{(k+1)T}{n}}-W_{t})\,|\,\mathcal{F}_t\right)\\
&=\frac{1}{P(\hat A_k|\mathcal F_{t})}\mathbb E_{\mathbb P}
\left(\mathbb I_{\hat A_k}(W_{\frac{(k+1)T}{n}}-W_{t})\,|\,\mathcal F_t\right)\\
&=0.
\end{align*}
Thus,
 for any $k<n$ the stochastic process
$\{W_t\}_{t=kT/n}^{(k+1)T/n}$ is a $\mathbb Q$-martingale, and so
$W=\{W_t\}_{t=0}^T$
is a $\mathbb Q$-martingale.
As $\mathbb Q\ll\mathbb P$, we conclude that
$\langle W\rangle _t\equiv t$, $\mathbb Q$-a.s. This together with L\'evy's characterization theorem
yields that
$W$ is a Brownian motion
with respect to $\mathbb Q$ and $\{\mathcal F_t\}_{t=0}^T$.

We arrive to the final step
of the proof.
Consider the event
$$\hat A:=\Big(\bigcap_{i=0}^{n-1} \hat A_i\Big) \cap \Big\{\max_{|u-v|\leq\frac{T}{n}} |\phi_u-\phi_v|
\leq \delta \Big\}. $$
The stochastic process $\phi$ is adapted to the filtration generated by $W$; in particular,
$\phi$ is determined by $\{W_t\}_{t=0}^T$.
 Hence ($W$ is a Brownian motion under $\mathbb P$ and $\mathbb Q$), the distributions
of $\phi$ under $\mathbb P$ and $\mathbb Q$ are the same.
This, together with (\ref{3.10+}) and the fact that
$\mathbb Q\big(\bigcap_{i=0}^{n-1} \hat A_i\big)=1$ yields
\begin{equation}\label{3.12}
\mathbb Q(\hat A)=\mathbb Q \Big(\max_{|u-v|\leq\frac{T}{n}} |\phi_u-\phi_v|
\leq \delta \Big)=\mathbb P \Big(\max_{|u-v|\leq\frac{T}{n}} |\phi_u-\phi_v|
\leq \delta \Big)\geq 1-\delta.
\end{equation}
Next, let $k<n$ and
$t\in [kT/n, (k+1)T/n]$. Observe that
$\phi_0=\ln \nu^{(2)}_0$. Thus, we have on the event $\hat A$
\begin{align*}
|\ln \nu_t-\ln\alpha_t|&=|\ln \nu^{(2)}_t-\phi_t|\\
&\leq
\sum_{i=0}^{k-1} \left|\ln\nu^{(2)}_{\frac{(i+1)T}{n}}-\ln\nu^{(2)}_{\frac{iT}{n}}-
\phi_{\frac{i T}{n}}+\phi_{\frac{(i-1)T}{n}}\right|\\
& \quad +|\phi_{\frac{k T}{n}}-\phi_{\frac{(k-1)T}{n}}|+
|\phi_t-\phi_{\frac{kT}{n}}|+|\ln\nu^{(2)}_{t}-\ln\nu^{(2)}_{\frac{kT}{n}}|\\
& \leq
\frac{\delta k}{n}+2\delta+\frac{\delta}{n}+\delta+(nt/T-k)\delta\\
&\leq 6\delta.
\end{align*}
From the inequality
$$|e^x-e^y|\leq e^{\max(x,y)}|x-y|\leq e^x e^{|x-y|}|x-y| \ \ x,y\in\mathbb R$$
we conclude that on the event
$\hat A$, (take $x=\ln \alpha_t$, $y=\ln \nu_t$)
$$\sup_{0\leq t\leq T}|\alpha_t-\nu_t|\leq 6\delta e^{6\delta}||\alpha||_{\infty}. $$
 This, together with applying (\ref{3.12})
for sufficiently small $\delta>0$ (recall that $\alpha$ is uniformly bounded) we get
$\mathbb Q\left(||\alpha-\nu||_{\infty}<\epsilon\right)>1-\epsilon,$
and the proof is completed.
\end{proof}

\section{Proof of Theorem \ref{thm2.1game}}
\label{sec:proof}
\setcounter{equation}{0}
In this section, we finish the proof of Theorem~\ref{thm2.1game} by showing that the inequality  (\ref{2.30}) holds true.
It suffices to show that for any super-replicating strategy $(\pi,\sigma)$
we have the inequality
\begin{equation}\label{new}
Z^\pi_0\geq g(S_0).
\end{equation}
To this end, let $(\pi,\sigma)$ be a super-replicating strategy.
Choose $\epsilon>0$. The stochastic process ${\{r_t\}}_{t=0}^T$
is uniformly bounded\replaced{, thus}{ and so,} there exists $\mathbf T<T$ such that
\begin{equation}\label{3.20}
\int_{0}^\mathbf T r_t\, dt<\epsilon.
\end{equation}
Let ${\{\mathcal{G}_t\}}_{t=0}^\mathbf T$
be the filtration generated by
$W$ and completed by the null sets.
Denote by ${\mathcal T}_{\mathbf T}$ the set of all stopping times with respect to the filtration
${\{\mathcal{G}_t\}}_{t=0}^\mathbf T$ with values in $[0,\mathbf T]$.
By Corollary \ref{4.new}, there exists a stochastic process
$\alpha\in\mathcal C(\nu_0)$ such that
\begin{equation}\label{3.21}
\inf_{\zeta\in{\mathcal T}_{\mathbf T}} \mathbb E_{\mathbb P}\Big[f_2(\mathbf S^{(\alpha)}_{\zeta})\mathbb{I}_{\zeta<\mathbf T}
+f_1(\mathbf S^{(\alpha)}_{\mathbf T})\mathbb{I}_{\zeta=\mathbf T}\Big]> g(S_0)-\epsilon\added{,}
\end{equation}
where
\begin{equation*}
\mathbf S^{(\alpha)}_t=S_0e^{\int_{0}^t \alpha_u dW_u-\frac{1}{2}\int_{0}^t \alpha^2_u du}, \ \ t\in [0,T].
\end{equation*}
Choose $\delta>0$.
The financial market is fully incomplete. Hence by definition, we obtain 
a probability measure $\mathbb Q \ll \mathbb P$ such that
\begin{equation}\label{3.23}
\mathbb Q\big(\|\alpha-\nu\|_{\infty}\geq\delta\big)<\delta
\end{equation}
and that $W$ is a \replaced{Brownian motion}{Wiener process} with respect to $\mathbb Q$ and $\{\mathcal F_t\}_{t=0}^T$.

Define the stopping time
$\tau=\mathbf \inf\{t:|\alpha_t-\nu_t|\replaced{\geq}{=}\delta\}\wedge \mathbf T$ and denote
$\pi=(Z^\pi_0,\gamma)$. From (\ref{2.4+})--(\ref{2.5}), it follows that the stochastic integral
$$\int_{0}^{t\wedge\sigma} \gamma_u \tilde S_u \nu_u \,dW_u, \ \ t\in [0,T]$$
is uniformly bounded from below, and so \replaced{it is}{its} a supermartingale
with respect to the probability measure $\mathbb Q$.
Thus, from (\ref{2.4+})--(\ref{2.5})
$$
\mathbb E_{\mathbb Q}\left[\frac{B_0}{B_{\sigma\wedge\tau}}\big(f_2(S_{\sigma})\mathbb{I}_{\sigma<\tau}+f_1(S_{\tau})\mathbb{I}_{\tau\leq\sigma}\big)\right]\leq
\mathbb E_{\mathbb Q}\left[\frac{B_0}{B_{\sigma\wedge\tau}}Z^\pi_{\sigma\wedge\tau}\right]\leq Z^\pi_0,$$
and so\deleted{,} from (\ref{3.20})\added{,} we conclude \added{that}
\begin{equation}\label{3.23+}
e^{\epsilon} Z^\pi_0 \geq
\mathbb E_{\mathbb Q}\big[f_2(S_{\sigma})\mathbb{I}_{\sigma<\tau}+f_1(S_{\tau})\mathbb{I}_{\tau\leq\sigma}\big].
\end{equation}
Clearly,
$$\int_{0}^{\sigma\wedge\tau}|\alpha^2_t-\nu^2_t|dt\leq \delta (2||\alpha||_{\infty}+\delta) T,$$
and from It\added{\^o}'s Isometry
$$\mathbb E_{\mathbb Q}\left[\Big(\int_{0}^{\sigma\wedge\tau}(\nu_t-\alpha_t) \,dW_t\Big)^2\right]\leq \delta^2 T.$$
Thus, from the Markov inequality we get \added{for sufficiently small $\delta$}
\begin{equation}\label{3.26}
\mathbb Q\left(\int_{0}^{\sigma\wedge\tau}|\alpha^2_t-\nu^2_t|\,dt+
\Big|\int_{0}^{\sigma\wedge\tau}(\nu_t-\alpha_t) \,dW_t\Big|>2\sqrt \delta\right)<c \sqrt{\delta}
\end{equation}
 for some constant $c>0$ (which may depend on the chosen $\epsilon>0$).
The SDE (\ref{2.2}) implies that
$$S_{\sigma\wedge\tau}=S_0e^{\int_{0}^{\sigma\wedge\tau} \nu_t dW_t+\int_{0}^{\sigma\wedge\tau} (r_t-\nu^2_t/2) dt}.$$
From (\ref{3.20}) and (\ref{3.26}) we get that for sufficiently small $\delta$
\begin{equation}\label{3.27}
\mathbb Q\left(|\ln S_{\sigma\wedge\tau}-\ln \mathbf S^{(\alpha)}_{\sigma\wedge\tau}|>2\epsilon\right)<c\sqrt \delta.
\end{equation}
Now, we arrive \replaced{at}{to} the final step of the proof.
Set\deleted{,}
$$\tilde\sigma=\sigma\wedge \mathbf T, \ \quad X=\sup_{0\leq t\leq T} f_2(\mathbf S^{(\alpha)}_t),$$
and introduce the event
$U=(\tau<\mathbf T)\cup (|\ln S_{\sigma\wedge\tau}-\ln \mathbf S^{(\alpha)}_{\sigma\wedge\tau}|>2\epsilon).$
\replaced{We deduce from}{From} (\ref{2.3}) \deleted{it follows} that
\begin{equation}\label{3.28}
|\ln x-\ln y|\leq 2\epsilon\Rightarrow f_i(y)\geq \frac{(1-L(e^{2\epsilon}-1))f_i(x)-L x(e^{2\epsilon}-1)}{1+L(e^{2\epsilon}-1)} , \ \ i=1,2.
\end{equation}
From (\ref{3.23+}) and (\ref{3.28}) we obtain
\begin{align}
e^{\epsilon} Z^\pi_0&\geq
\mathbb E_{\mathbb Q}\big[\mathbb{I}_{\Omega\setminus U}\left(f_2(S_{\tilde\sigma})
\mathbb{I}_{\tilde\sigma<\mathbf T}+f_1(S_{\mathbf T})\mathbb{I}_{\tilde\sigma=\mathbf T}\right)\big] \label{3.29}\\
&\geq
\frac{1-L(e^{2\epsilon}-1)}{1+L(e^{2\epsilon}-1)} \,\mathbb E_{\mathbb Q}\left[\mathbb{I}_{\Omega\setminus U}\left(f_2(\mathbf S^{(\alpha)}_{\tilde\sigma})
\mathbb{I}_{\tilde\sigma<\mathbf T}+f_1(\mathbf S^{(\alpha)}_{\mathbf T})\mathbb{I}_{\tilde\sigma=\mathbf T}\right)\right] \nonumber\\
& \phantom{ \geq } \ -
\frac{L (e^{2\epsilon}-1)}{1+L(e^{2\epsilon}-1)}\,\mathbb E_{\mathbb Q}[\mathbf S^{(\alpha)}_{\tilde\sigma}] \nonumber\\
&\geq
\frac{1-L(e^{2\epsilon}-1)}{1+L(e^{2\epsilon}-1)}\,\mathbb E_{\mathbb Q}\left[f_2(\mathbf S^{(\alpha)}_{\tilde\sigma})
\mathbb{I}_{\tilde\sigma<\mathbf T}+f_1(\mathbf S^{(\alpha)}_{\mathbf T})\mathbb{I}_{\tilde\sigma=\mathbf T}\right] \nonumber\\
&\phantom{ \geq } \ -
\frac{1-L(e^{2\epsilon}-1)}{1+L(e^{2\epsilon}-1)}\,\mathbb E_{\mathbb Q}[\mathbb{I}_U X]-
\frac{L S_0(e^{2\epsilon}-1)}{1+L(e^{2\epsilon}-1)}.\nonumber
\end{align}
The growth condition (\ref{2.4}) implies
that $\mathbb E_{\mathbb Q} [X^2]<\infty$.
Observe that $(\tau<\mathbf T)\subset (\|\alpha-\nu\|_{\infty}\geq \delta)$.
Thus, from the Cauchy-–Schwarz inequality, (\ref{3.23}) and (\ref{3.27}), we get that for sufficiently small $\delta>0$
\begin{equation}\label{3.30}
\mathbb E_{\mathbb Q}[\mathbb{I}_U X]\leq \left(\mathbb E_{\mathbb Q} [X^2]\right)^{1/2}\left(\delta+c\sqrt\delta\right)^{1/2}<\epsilon.
\end{equation}
Finally, we estimate $\mathbb E_{\mathbb Q}\left[f_2(\mathbf S^{(\alpha)}_{\tilde\sigma})
\mathbb{I}_{\tilde\sigma<\mathbf T}+f_1(\mathbf S^{(\alpha)}_{\mathbf T})\mathbb{I}_{\tilde\sigma=\mathbf T}\right]$.
Denote by $\mathbb T$ the set of all stopping times with respect to the filtration
$\{\mathcal F_t\}_{t=0}^T$ with values in $[0,\mathbf T]$.
The stochastic process $W$ is a Brownian motion under the probability measure $\mathbb Q$
and the filtration ${\{\mathcal F_t\}}_{t=0}^T$.
Thus\added{,} from the Markov property of Brownian motion,
the fact that $\alpha$ is adapted to the filtration $\{\mathcal G_t\}_{t=0}^T$,
$\tilde\sigma\in\mathbb T$ and (\ref{3.21}), it follows that
\begin{align*}
\mathbb E_{\mathbb Q}\left[f_2(\mathbf S^{(\alpha)}_{\tilde\sigma})
\mathbb{I}_{\tilde\sigma<\mathbf T}+f_1(\mathbf S^{(\alpha)}_{\mathbf T})\mathbb{I}_{\tilde\sigma=\mathbf T}\right]\geq
&\inf_{\zeta\in{\mathbb T}} \mathbb E_{\mathbb Q}\Big[f_2(\mathbf S^{(\alpha)}_{\zeta})\mathbb{I}_{\zeta<\mathbf T}
+f_1(\mathbf S^{(\alpha)}_{\mathbf T})\mathbb{I}_{\zeta=\mathbf T}\Big]\\=
&\inf_{\zeta\in\mathcal T_{\mathbf T}} \mathbb E_{\mathbb Q}\Big[f_2(\mathbf S^{(\alpha)}_{\zeta})\mathbb{I}_{\zeta<\mathbf T}
+f_1(\mathbf S^{(\alpha)}_{\mathbf T})\mathbb{I}_{\zeta=\mathbf T}\Big]\\=
&\inf_{\zeta\in{\mathcal T}_{\mathbf T}} \mathbb E_{\mathbb P}\Big[f_2(\mathbf S^{(\alpha)}_{\zeta})\mathbb{I}_{\zeta<\mathbf T}
+f_1(\mathbf S^{(\alpha)}_{\mathbf T})\mathbb{I}_{\zeta=\mathbf T}\Big]\\
> & \,g(S_0)-\epsilon.
\end{align*}
This together with (\ref{3.29})--(\ref{3.30}) gives
$$e^{\epsilon} Z^\pi_0\geq \frac{1-L (e^{2\epsilon}-1)}{1+L(e^{2\epsilon}-1)}(g(S_0)-\epsilon)-
\frac{1-L (e^{2\epsilon}-1)}{1+L(e^{2\epsilon}-1)}\epsilon-\frac{L S_0(e^{2\epsilon}-1)}{1+L(e^{2\epsilon}-1)},
$$
and by letting $\epsilon\downarrow 0$ we obtain
(\ref{new}).
\qed
\begin{rem}\label{rem.pathdependent}
A natural question is whether for game options with path dependent payoffs the model free super--replication price
is equal to the price achieved in fully incomplete markets (see Remark~\ref{rem4.3}). In order to answer this question
we should develop a dual characterization for the super--replication price of
path dependent game options in a model free setup. This was not done so far.
\end{rem}
\section{Auxiliary lemmas for the proof of Theorem~\ref{thm2.1game}}
\label{sec:lemmas}
The goal of this section is to establish Corollary \ref{4.new} which provides a connection between the
function $g$ (which is the game variant of a concave envelope) and the left hand side of
(\ref{4.nnn})
which can be viewed as an optimal stopping problem under volatility uncertainty.

Consider the probability space
$(\Omega,{\mathcal{F}},{\mathbb P})$ and
the filtration $\{\mathcal G_t\}_{t=0}^T$
generated by the \replaced{Brownian motion}{Wiener process} $\{W_t\}_{t=0}^T$\replaced{,}{and} completed by the null sets.
For any $u\in [0,T]$ we denote by $\mathcal T_u$ the set of all stopping times
with respect to the filtration $\{\mathcal G_t\}_{t=0}^T$
with values in $[0,u]$.
For any $x>0$ and any \added{(sufficiently integrable)} progressively
measurable process  $\alpha=\{\alpha_t\}_{t=0}^T$ (with respect to $\{\mathcal G_t\}_{t=0}^T$)
define the process
$$\mathbf S^{\alpha,x}_t=xe^{\int_{0}^t \alpha_v dW_v-\frac{1}{2}\int_{0}^t \alpha^2_v dv}, \ \ t\in [0,T].$$
Denote by $\mathcal{A}$ the set of all non--negative, progressively measurable processes
$\alpha=\{\alpha_t\}_{t=0}^T$ \added{with $\int_0^T \alpha^2_t \,dt <\infty$ \, a.s.} which satisfy the following: there exists a constant $C=C(\alpha)$ such that
$\frac{1}{C}\leq \replaced{\mathbf S^{\alpha,1}}{S^{\alpha,1}}\leq C$.
Define the function $G:(0,\infty)\times (0,T]\rightarrow\mathbb R$
\begin{equation}\label{4.0}
G(x,u):=\sup_{\alpha\in\mathcal A}\inf_{\zeta\in\mathcal{T}_{u}}
 \mathbb E_{\mathbb P}\left[f_2(\mathbf S^{\alpha,x}_{\zeta})\mathbb{I}_{\zeta<u}
+f_1(\mathbf S^{\alpha,x}_{u})\mathbb{I}_{\zeta=u}\right].
\end{equation}
The following lemma is similar to Dolinsky (2013, Lemmas 4.1--4.2). As the present setup is a bit different, we provide
for reader's convenience a self contained proof.
\begin{lem}\label{lem4.0}${}$\\
i.The function $G(x,u)$ does not depend on $u$\replaced{,}{.} i.e. for all $u<T$
$$G(x,u)=G(x,T).$$
ii. The function $G(x):=G(x,T)$ is continuous and satisfies $f_1\leq G\leq f_2$.\\
iii. The function $G(x)$ is concave in every interval in which $G<f_2$.
\end{lem}
\begin{proof}
i. The proof will be done by a standard time scaling \added{argument}.
Let $x>0$ and $u\in (0,T]$.
Consider the Brownian motion \replaced{defined}{given} by
$\hat{\mathbf W}_t:=\sqrt\frac{u}{T}W_{\frac{tT}{u}}$, $t\in [0,u]$.
Let $\{\hat{\mathcal{G}}_t\}_{t=0}^{u}$ be the filtration which is generated by $\{\hat{\mathbf W}_t\}_{t=0}^u$ (completed by the null sets)
and let $\hat{\mathcal T}_u$ be the
set of all \added{$\{\hat{\mathcal{G}}_t\}_{t=0}^{u}$--}stopping times with values in $[0,u]$\deleted{with respect to this filtration}.
For any $x>0$ and any \added{$\{\hat{\mathcal{G}}_t\}_{t=0}^{u}$--}progressively
measurable \added{(sufficiently integrable)} process \deleted{(with respect to $\{\hat{\mathcal G}_t\}_{t=0}^u$)} $\hat\alpha=\{\hat\alpha_t\}_{t=0}^u$
define the process
$$\replaced{\hat{\mathbf{S}}^{\hat\alpha,x}_t}{\mathbf{S}^{\hat\alpha,x}_t}=xe^{\int_{0}^t \hat\alpha_v d\hat{\mathbf W}_v-\frac{1}{2}\int_{0}^t \hat\alpha^2_v dv}, \ \ t\in [0,u].$$
Denote by $\hat{\mathcal{A}}$ the set of all non--negative, \added{$\{\hat{\mathcal{G}}_t\}_{t=0}^{u}$-}progressively measurable processes
$\hat\alpha=\{\hat\alpha_t\}_{t=0}^u$ \added{with $\int_0^u \hat{\alpha}^2_t \,dt<\infty$ a.s.} for which
there exists a constant $C=C(\hat\alpha)$ such that
$\frac{1}{C}\leq \mathbf{\hat{S}}^{\hat{\alpha},1}\leq C$.
Observe that the maps
$\phi:\mathcal{T}_T\rightarrow \hat{\mathcal T}_u$
and $\psi:\mathcal{A}\rightarrow \hat{\mathcal A}$
given by $\phi(\zeta):=\frac{\zeta u}{T}$
and $[\psi(\alpha)]_t:=\sqrt\frac{T}{u}\alpha_{\frac{t T}{u}}$\replaced{,}{.} $t\in [0,u]$\added{,}
are bijections.
Moreover,
$\mathbf S^{\alpha,x}_t=\hat{\mathbf S}^{\psi(\alpha),x}_{\phi(t)}$, $t\in [0,T]$.
Thus, we obtain
\begin{align*}
G(x,T)&=\sup_{\hat\alpha\in\hat{\mathcal A}}\inf_{\hat\zeta\in\hat{\mathcal{T}}_{u}}
 \mathbb E_{\mathbb P}\left[f_2(\hat{\mathbf S}^{\hat\alpha,x}_{\hat\zeta})\mathbb{I}_{\hat\zeta<u}
+f_1(\hat{\mathbf S}^{\hat\alpha,x}_{u})\mathbb{I}_{\hat\zeta=u}\right]\\
&=
\sup_{\alpha\in\mathcal A}\inf_{\zeta\in\mathcal{T}_{u}}
 \mathbb E_{\mathbb P}\left[f_2({\mathbf S}^{\alpha,x}_{\zeta})\mathbb{I}_{\zeta<u}
+f_1({\mathbf S}^{\alpha,x}_{u})\mathbb{I}_{\zeta=u}\right]=G(x,u)\replaced{,}{,}
\end{align*}
as required. \qed\\
ii. In (\ref{4.0}), if we put $\zeta\equiv 0$ we obtain
$G\leq f_2$ and for $\alpha\equiv 0$ we obtain $G\geq f_1$.
Thus, $f_1\leq G\leq f_2$.
Next, we prove the continuity of $G$.
Let $x,y>0$. Denote
$z=\replaced{\max\big(\frac{x}{y},\frac{y}{x}\big)}{\max\left(\frac{x^2}{y^2},\frac{y^2}{x^2}\right)}.$
 Similarly to (\ref{3.28})\added{,} we obtain
that for any $\alpha\in\replaced{\mathcal{A}}{\Gamma}$ and $t\in [0,T]$
$$
f_i(\mathbf S^{\alpha,y}_t)\geq \frac{(1-L(z-1))f_i(\mathbf S^{\alpha,x}_t)-L \mathbf S^{\alpha,x}_t(z-1)}{1+L(z-1)} , \ \ i=1,2.
$$
This together with the fact that $\{\mathbf S^{\alpha,x}_t\}_{t=0}^T$ is a supermartingale gives
$$G(y)\geq \frac{(1-L(z-1))G(x)- L x(z-1)}{1+L(z-1)}.$$
As $x,y$ are arbitrary we\deleted{e} conclude
that $G(y)\geq\limsup_{n\rightarrow\infty} G(x_n)$ for any sequence $x_n\rightarrow y$,
\replaced{which}{this} yields the upper semi--continuity. Similarly,  for any sequence
$y_n\rightarrow x$ \replaced{we have}{,} $G(x)\leq \liminf_{n\rightarrow\infty} G(y_n)$, \replaced{which}{this} yields the lower semi--continuity and completes the proof.\qed\\
iii. Let $D\added{\subseteq(0,\infty)}$ be an open interval such that $G<f_2$ in $D$.
Fix $x_1,x_2, x_3\in D$ and assume that $0<x_2<x_3<x_1$.
Let $0<\lambda<1$ such that
$x_3=\lambda x_1+(1-\lambda)x_2$.
We need to show that
\begin{equation}\label{4.concave}
G(x_3)\geq \lambda G(x_1)+(1-\lambda) G(x_2).
\end{equation}
Let $a\in\mathbb R$ be a constant such that
$\mathbb P(W_{\frac{T}{2}}>a)=\lambda.$
Define the martingale
$$M_t=
\mathbb{E}_{\mathbb P}\Big[x_1\mathbb{I}_{W_{\frac{T}{2}}>a}+x_2\mathbb{I}_{W_{\frac{T}{2}}<a}\,\Big|\,\mathcal{G}_t\Big], \ \ t\in [0,T/2].$$
Observe that $M_0=x_3$. \replaced{We deduce from}{From} It\^o's formula\deleted{it follows} that
\begin{equation*}
M_t=x_3e^{\int_{0}^t \alpha_v\, dW_v-\frac{1}{2}\int_{0}^t \alpha^2_v \,dv},
\end{equation*}
where \added{for $t<\frac{T}{2}$}
\begin{align*}
\alpha_t&=\frac{1}{M_t}\frac{\partial }{\partial W_t}\left(x_1\int_{a-W_t}^{\infty}
\frac{\exp\left(-\frac{v^2}{2(\frac{T}{2}-t)}\right)}{\sqrt{2\pi(\frac{T}{2}-t)}}dv+x_2
\int_{-\infty}^{a-W_t}
\frac{\exp\left(-\frac{v^2}{2(\frac{T}{2}-t)}\right)}{\sqrt{2\pi(\frac{T}{2}-t)}}
dv\right)\nonumber\\
&=\frac{x_1-x_2}{M_t} \frac{\exp\left(-\frac{(a-W_t)^2}{2(\frac{T}{2}-t)}\right)}{\sqrt{2\pi(\frac{T}{2}-t)}}>0,
\end{align*}
\replaced{and for}{For} $t=\frac{T}{2}$ define $\alpha_{\frac{T}{2}}\equiv 0$.

Next, choose $\epsilon>0$.
There exist\deleted{s} $\alpha^{(1)},\alpha^{(2)}\in \mathcal{A}$ such that
\begin{equation}\label{4.5}
G(x_i)<\epsilon+\inf_{\zeta\in\mathcal{T}_{\frac{T}{2}}}
 \mathbb E_{\mathbb P}\left[f_2(\mathbf S^{\alpha^{(i)},x_i}_{\zeta})\mathbb{I}_{\zeta<\frac{T}{2}}
+f_1(\mathbf S^{\alpha^{(i)},x_i}_{\frac{T}{2}})\mathbb{I}_{\zeta=\frac{T}{2}}\right], \ \ i=1,2.
\end{equation}
The processes $\alpha^{(i)}$ are progressively measurable with respect to the filtration $\{\mathcal{G}_t\}_{t=0}^{T}$ and so
there exist\deleted{s} progressively measurable maps
$\varphi_i:C[0,T]\rightarrow C_{+}[0,T]$
(i.e. $[\varphi_i(y)]_{[0,t]}$ depends only on $y_{[0,t]}$)
such that
$\alpha^{(i)}=\varphi\added{_i}\deleted{^{(i)}}(W)$ a.s..
Consider the Brownian motion $\mathbf W_t=W\added{^{(2)}}_{t+\frac{T}{2}}-W\added{^{(2)}}_{\frac{T}{2}}$, \replaced{$0\leq t \leq \frac{T}{2}$}{$t\geq 0$}.
We extend the process $\alpha$ to the interval $(T/2,T]$ by \added{setting}
$$\alpha_{t+\frac{T}{2}}=\mathbb{I}_{W_{\frac{T}{2}}>a}[\varphi_1(\mathbf W)]_t+\mathbb{I}_{W_{\frac{T}{2}}\replaced{\leq}{<}a}[\varphi_2(\mathbf W)]_t, \ \ 0<t\leq \frac{T}{2}.$$
Clearly, the process ${\{\alpha_t\}}_{t=0}^T$ is non--negative and progressively measurable
with respect to the filtration $\{\mathcal{G}_t\}_{t=0}^{T}$.
The martingale $\{M_t\}_{t=0}^{T/2}$ satisfies
$0< \replaced{x_2}{x_1}\leq M\leq\replaced{x_1}{x_2}$. This together with the fact that $\alpha^{(1)},\alpha^{(2)}\in\mathcal A$
yields that
$\alpha\in\mathcal{A}$. Thus,
\begin{equation}\label{4.6}
G(x_3)\geq \inf_{\zeta\in\mathcal{T}_{T}}
 \mathbb E_{\mathbb P}\left[f_2(\mathbf S^{\alpha,x_3}_{\zeta})\mathbb{I}_{\zeta<T}
+f_1(\mathbf S^{\alpha,x_3}_{T})\mathbb{I}_{\zeta=T}\right].
\end{equation}

Now, we use that $G<f_2$ in $D$. Define the \deleted{stochastic}process\deleted{es}
\begin{equation*}
Z_t=\essinf_{\zeta\in\mathcal{T}_{[0,T]},\zeta\geq t}
\mathbb E_{\mathbb P}\left[f_2(\mathbf S^{\alpha,x_3}_{\zeta})\mathbb{I}_{\zeta<T}
+f_1(\mathbf S^{\alpha,x_3}_{T})\mathbb{I}_{\zeta=T}\, \Big| \,\mathcal G_t\right], \ \  t\in [0,T],
\end{equation*}
and the stopping time $\eta\in\mathcal{T}_{[0,T]}$ by,
\begin{equation*}
\eta=\inf\{t: Z_t=f_2(\mathbf S^{\alpha,x_3}_t)\}\wedge T.
\end{equation*}
From the general theory of optimal stopping
(see Peskir and Shiryaev 2006, chapter I)\added{,} it follows
that
\begin{equation*}
Z_0=\mathbb E_{\mathbb P} \Big[f_2(\mathbf S^{\alpha,x_3}_{\eta})\mathbb{I}_{\eta<T}
+f_1(\mathbf S^{\alpha,x_3}_{T})\mathbb{I}_{\eta=T}\Big].
\end{equation*}
The strong Markov property of Brownian \added{motion} implies that \added{for $t< \frac{T}{2}$}
$$Z_t\leq G(\mathbf S^{\alpha,x_3}_{t},T-t)=G(M_t)<f_2(M_t)\added{,}$$
where the last inequality follows from the fact that $M_t\in D$.
We conclude that $\eta\geq \frac{T}{2}$ a.s., and so
from the independence of $\{\mathbf W_t\}_{t=0}^{T/2}$ and $\{W\added{^{(2)}}_t\}_{t=0}^{T/2}$
\begin{align*}
Z_0=\mathbb E_{\mathbb P}Z_{\frac{T}{2}}=& \
\lambda \inf_{\zeta\in\mathcal{T}_{\frac{T}{2}}}
 \mathbb E_{\mathbb P}\left[f_2(\mathbf S^{\alpha^{(1)},x_1}_{\zeta})\mathbb{I}_{\zeta<\frac{T}{2}}
+f_1(\mathbf S^{\alpha^{(1)},x_1}_{\frac{T}{2}})\mathbb{I}_{\zeta=\frac{T}{2}}\right]\\
&\ +
(1-\lambda) \inf_{\zeta\in\mathcal{T}_{\frac{T}{2}}}
 \mathbb E_{\mathbb P}\left[f_2(\mathbf S^{\alpha^{(2)},x_2}_{\zeta})\mathbb{I}_{\zeta<\frac{T}{2}}
+f_1(\mathbf S^{\alpha^{(2)},x_2}_{\frac{T}{2}})\mathbb{I}_{\zeta=\frac{T}{2}}\right].
\end{align*}
This together with (\ref{4.5})--(\ref{4.6})
yields
$$G(x_3)\geq Z_0\geq \lambda G(x_1)+(1-\lambda) G(x_2)-\epsilon\added{,}$$
and by letting $\epsilon\downarrow 0$ we get $(\ref{4.concave})$, which completes the proof.
\end{proof}
Recall the set $\mathcal{C}(\nu_0)$, which was introduced in the beginning of Section~\ref{sec:2}, namely the set of all
continuous, strictly positive
stochastic processes $\alpha={\{\alpha_t\}}_{t=0}^T$
which are adapted with respect to the filtration generated by $W$ completed by the null sets, and satisfy:
i. $\alpha_0=\nu_0$. ii.
$\alpha$ and $\frac{1}{\alpha}$ are uniformly bounded.
Define the function $F:(0,\infty)\times (0,T]\rightarrow \mathbb R$ by
$$F(x,u)=\sup_{\alpha\in\mathcal{C}(\nu_0) }\inf_{\zeta\in\mathcal{T}_{u}}
 \mathbb E_{\mathbb P}\left[f_2(\mathbf S^{\alpha,x}_{\zeta})\mathbb{I}_{\zeta<u}
+f_1(\mathbf S^{\alpha,x}_{u})\mathbb{I}_{\zeta=u}\right].$$
\begin{lem}\label{lem4.1}
For any $x>0$ and $u\in (0,T]$, $F(x,u)\geq G(x)$.
\end{lem}
\begin{proof}
Fix $x>0$, $u\in(0,T]$ and choose $\epsilon>0$. Let $\alpha\in\mathcal{A}$ such that
\begin{equation}\label{4.100}
G(x)<\epsilon+\inf_{\zeta\in\mathcal{T}_{u}}
 \mathbb E_{\mathbb P}\left[f_2({\mathbf S}^{\alpha,x}_{\zeta})\mathbb{I}_{\zeta<u}
 +f_1({\mathbf S}^{\alpha,x}_{u})\mathbb{I}_{\zeta=u}\right].
\end{equation}
Notice that $d\mathbf{S}^{\alpha,x}_t=\alpha_t \mathbf{S}^{\alpha,x}_t dW_t$, and so from the fact that
$\frac{1}{C}\leq \mathbf{S}^{\alpha,x}\leq C$ for some constant $C$, we deduce that
$\mathbb E_{\mathbb P}\big[\int_{0}^u\alpha^2_t \,dt\big]<\infty$.
Thus,
by applying standard density arguments, it follows that we can find a sequence of stochastic processes
$\replaced{(\alpha^{(n)})\subseteq}{\alpha^{(n)}\in}\mathcal{C}(\nu_0)$ such that
$$\lim_{n\rightarrow\infty}\mathbb E_{\mathbb P}\left[\int_{0}^u\big((\alpha^{(n)}_t-\alpha_t)^2+
|(\alpha^{(n)}_t)^2-(\alpha_t)^2|\big)\,dt\right]=0.$$
\replaced{We deduce from}{From} the \replaced{Burkholder--Davis--Gundy inequality}{Doob--Kolmogorov inequality and Ito's Isometry}\deleted{it follows} that
$$\lim_{n\rightarrow\infty}\mathbb E_{\mathbb P}\left[\sup_{0\leq v\leq u}\left(\int_{0}^v (\alpha^{(n)}_t-\alpha_t) \,dW_t\right)^2\right]=0.$$
\replaced{Therefore, we}{We} conclude the following convergence
\begin{equation*}
\sup_{0\leq t\leq u}|\ln{\mathbf S}^{\alpha_n,x}_t-\ln\mathbf{S}^{\alpha,x}_t|\rightarrow 0 \  \mbox{in}  \ \mbox{probability}.
\end{equation*}
Next, choose $\delta>0$. There exists $n\in\mathbb N$ such that
$$\mathbb P\left(\sup_{0\leq t\leq u}|\ln{\mathbf S}^{\alpha_n,x}_t-\ln\mathbf{S}^{\alpha,x}_t|>\delta\right)<\delta.$$
Set $X=\sup_{0\leq t\leq u} f_2(\mathbf{S}^{\alpha,x}_t)$ and the event
$U=\big(\sup_{0\leq t\leq u}|\ln{\mathbf S}^{\alpha_n,x}_t-\ln\mathbf{S}^{\alpha,x}_t|>\delta\big)$.
The growth condition (\ref{2.4}) implies
that $\mathbb E_{\mathbb P} [X^2]<\infty$.
Similarly to (\ref{3.28})--(\ref{3.29}), we get
\begin{align*}
F(x\added{,u})&\geq \inf_{\zeta\in\mathcal{T}_{u}}
 \mathbb E_{\mathbb P}\left[\mathbb{I}_{\Omega\setminus U}\left(f_2({\mathbf S}^{\alpha^{(n)},x}_{\zeta})\mathbb{I}_{\zeta<u}
 +f_1({\mathbf S}^{\alpha^{(n)},x}_{u})\mathbb{I}_{\zeta=u}\right)\right]\\
 &\geq
\frac{1-L(e^{2\delta}-1)}{1+L(e^{2\delta}-1)}\inf_{\zeta\in\mathcal{T}_{u}}\mathbb E_{\mathbb P}
\left[f_2({\mathbf S}^{\alpha,x}_{\zeta})\mathbb{I}_{\zeta<u}
+f_1({\mathbf S}^{\alpha,x}_{u})\mathbb{I}_{\zeta=u}\right]\\
& \phantom{ \geq }-
\frac{1-L(e^{2\delta}-1)}{1+L(e^{2\delta}-1)}\mathbb E_{\mathbb P}[X\mathbb{I}_U]-
\frac{L (e^{2\delta}-1)}{1+L(e^{2\delta}-1)}
\sup_{\zeta\in\mathcal{T}_{u}}\mathbb E_{\mathbb P}[{\mathbf S}^{\alpha,x}_{\zeta}]\\
&\geq
\frac{1-L(e^{2\delta}-1)}{1+L(e^{2\delta}-1)}\left(G(x)-\epsilon-\sqrt{\delta \mathbb E_{\mathbb P} [X^2]}\right)-
\frac{L (e^{2\delta}-1)}{1+L(e^{2\delta}-1)}x
\end{align*}
where the last inequality follows
from (\ref{4.100}), the Cauchy--Schwarz inequality and the fact that ${\mathbf S}^{\alpha,x}$ is a supermartingale.
By letting $\delta\downarrow 0$ we obtain
$F(x)\geq G(x)-\epsilon$, and by letting $\epsilon\downarrow 0$ we complete the proof.
\end{proof}
Next, recall the terms $\mathbb H$ and $g$ which
were defined before Assumption \ref{asm2.2}.
From Lemma \ref{lem4.0}, we conclude that $G\in\mathbb H$\replaced{,}{ and} in particular
$G\geq g$. This together with Lemma \ref{lem4.1} gives the following immediate corollary.
\begin{cor}\label{4.new}
For any $x>0$ and $u\in (0,T]$,
\begin{equation}\label{4.nnn}
\sup_{\alpha\in\mathcal{C}(\nu_0) }\inf_{\zeta\in\mathcal{T}_{u}}
 \mathbb E_{\mathbb P}\left[f_2(\mathbf S^{\alpha,x}_{\zeta})\mathbb{I}_{\zeta<u}
+f_1(\mathbf S^{\alpha,x}_{u})\mathbb{I}_{\zeta=u}\right]\geq g(x).
\end{equation}
\end{cor}

We end with the following remark.

\begin{rem}
Let us take $r\equiv 0$. Then by following the proof of (\ref{2.30})
and applying Lemmas \ref{lem4.0}--\ref{lem4.1}, we get that for any $u\in [0,T]$
$$V\geq F(S_0,u)\geq G(S_0)\geq g(S_0).$$ This together with the inequality
$V\leq \mathbf V\leq g(S_0)$ (Assumption \ref{asm2.2} holds true) gives
$F(S_0,u)=G(S_0)=g(S_0)$, i.e., we conclude that $F(x,u)=G(x)=g$ and $G$ is the minimal element in $\mathbb H$.
Observe that the functions $F,G$ are independent of the
interest rates, and so this result can be viewed as a general conclusion which provides a link between the
game variant of concave envelope $g$ and the value $G$ of the
optimal stopping problem under volatility uncertainty.
\end{rem}

\section{Density Results for martingale measures}\label{sec:8}\setcounter{equation}{0}
Recall the filtered probability space $(\Omega,\mathcal{F},\{\mathcal{F}_t\}_{t=0}^T, \mathbb{P})$
and the price process $S=\{S_t\}_{t=0}^T$
introduced in \eqref{2.2}.
For any probability measure $\mathbb Q\ll \mathbb P$,
we denote by
$\mathbb Q^S$ the distribution of the discounted stock price process
$\tilde S_t=\frac{S_t}{B_t}$, $t\in [0,T]$
on the canonical space $C[0,T]$. Namely,
$\mathbb Q^S(\mathbb A)=\mathbb Q (\tilde S\in \mathbb A)$ for any Borel
set $\mathbb A\in C[0,T]$.

Define
$\mathcal M^S=\{\mathbb Q^S: \mathbb Q\in\mathcal Q\}$, where  $\mathcal Q$ is the
set of all probability measures $\mathbb Q\ll \mathbb P$ such that $\{W_t\}_{t=0}^T$ is a
Brownian motion with respect to $\mathbb Q$ and
the filtration $\{\mathcal F_t\}_{t=0}^T$, as defined in Section \ref{sec:3+}.
Clearly, $\mathcal M^S\subset \mathcal M$, where $\mathcal M$ denotes the set of all strictly positive local martingale measures as in Section \ref{sec:3+}.
\begin{lem}\label{lem.density}
If the financial market given by \eqref{2.1}--\eqref{2.2} is fully incomplete, then
$\mathcal M^S$ is a weakly dense subset of $\mathcal M$.
\end{lem}
\begin{proof}
${}$\\
\textbf{First Step:}
Denote by $\mathcal M^b$ the set of all probability measures $\hat Q\in\mathcal M$ such that the
canonical process $\mathbb S$ is a $\hat Q$-martingale which satisfies
$\frac{1}{C}\leq\mathbb S\leq C$ $\hat Q$-a.s.
for some constant $C>0$ (which depends on $\hat Q$).
Let us show that $\mathcal M^b$ is a weakly dense subset of $\mathcal M$.
Let $Q\in\mathcal M$.
For any $C>0$ define the stopping time
$\tau_C=T\wedge\min\{t: \mathbb S_t\leq\frac{1}{C} \ \mbox{or} \ \mathbb S_t\geq C\}.$
Observe that the continuity of $\mathbb S$ implies that $\tau_C$ is a
stopping time with respect to the canonical filtration
$\mathbb F_t=\sigma\{\mathbb S_u:u\leq t\}$.
Consider the truncated stochastic process $\mathbb S^C$ given by
$\mathbb S^C_t=\mathbb S_{t\wedge\tau_C}$ , $t\in [0,T]$.
Let $Q^C$ be a probability measure on $C[0,T]$
defined by
$Q^C(\mathbb A)= Q (\mathbb S^C\in \mathbb A)$, for any Borel
set $\mathbb A\in C[0,T]$. Observe that $Q^C$ is the distribution of the process
$\mathbb S^C$ under the probability measure $Q$.
Clearly,
$\lim_{C\rightarrow\infty}\max_{0\leq t\leq T}|\mathbb S^C_t-\mathbb S_t|=0$ $Q$-a.s.
Hence, as $C\rightarrow\infty$, $Q^C$ converges weakly to $Q$.

From the Doob optional stopping theorem, see, e.g., Liptser and Shiryaev (2001, Theorem~3.6), it follows that under the probability measure
$Q$ the stochastic process $\mathbb S^C$
is a continuous martingale which  satisfies
$\frac{1}{C}\leq \mathbb S^C\leq C$ $Q$-a.s. Thus,
for any $C>0$, we have $Q^C\in\mathcal M^b$, so
we conclude that
$Q$ is a cluster point of $\mathcal M^b$, as required. \\
\textbf{Second Step:}
 Choose $Q\in\mathcal M^b$ and fix $\epsilon>0$.
There exists $n\in\mathbb N$ such that
 \begin{equation}\label{A.1}
 \mathbb E_{Q}\Big(\sup_{|u-v|\leq T/n} |{\mathbb S}_u-{\mathbb S}_v|\Big)<\epsilon.
 \end{equation}
From the existence of the regular distribution function
(see e.g. Shiryaev (1984, page 227)),
there exists for any $1\leq k<n$ a function
$\rho_k:\mathbb{R}\times\mathbb{R}^{k-1}\rightarrow [0,1]$
such that for any $y_1,...,y_{k-1}\in\mathbb R^{k-1}$,
$\rho_k(\cdot,y_1,...,y_{k-1})$,
is a distribution function on $\mathbb{R}$, and for any
$y$, $\rho_k(y,\cdot):\mathbb{R}^{k-1}\rightarrow [0,1]$
is measurable satisfying
$$
 Q\left( {\mathbb S}_{\frac{kT}{n}}\leq y\,\big|\,{\mathbb S}_{\frac{1}{n}},..., {\mathbb S}_{\frac{(k-1)T}{n}}
\right)
=\rho_k\left(y, {\mathbb S}_{\frac{1}{n}},..., {\mathbb S}_{\frac{(k-1)T}{n}}\right),
\ \  Q\mbox{-a.s.}
$$

Recall the probability space
$(\Omega,\mathcal F,\{\mathcal F_t\}_{t=0}^T,\mathbb P)$ and the
filtration $\{\mathcal G_t\}_{t=0}^T$ generated by $W$, completed by the $\mathbb P$-null sets.
Set $\tilde Z_i=W_{\frac{iT}{n}}-W_{\frac{(i-1)T}{n}}$, $i=1,...,n$.
Define
recursively the random variables
\begin{equation}\label{A.1+}
M_0=s  \ \ \mbox{and} \  \mbox{for} \  1\leq k \leq n \ \
M_k=\sup\{y\,|\,\rho_k(y,M_1,...,M_{k-1})<\Phi(\tilde Z_k)\}
\end{equation}
where $\Phi$ is the cumulative distribution function of
$\sqrt\frac{T}{n} W_1$.
As $\rho_k$ is a right-continuous non-decreasing function in the first variable,
we obtain that
$\{M_k\leq x\}=\{\rho_k(x,M_1,...,M_{k-1})\geq \Phi(\tilde Z_k) \}$.
Thus (by induction), we conclude that
$M_0,...,M_{n}$ are measurable. Moreover, as
$\Phi(\tilde Z_k)$ is a random variable uniformly distributed on $[0,1]$, we get
$$
 \mathbb P(M_k\leq y\,|\,M_1,...,M_{k-1})=
\rho_k(y, M_1,...,M_{k-1}).$$
Therefore, the joint distribution
of $M_0,...,M_n$ under $\mathbb P$ equals the joint distribution
of ${\mathbb S}_0,{\mathbb S}_{\frac{T}{n}},...,{\mathbb S}_T$ under $Q$.
 In particular, we have
\begin{equation}\label{bb}
\frac{1}{C}\leq M_n\leq C \quad \mathbb P\mbox{-a.s.}
\end{equation}
for some constant $C$.
Furthermore, there is for any $k$
a measurable function $g_k:\mathbb R^k\rightarrow\mathbb R$ such that
$M_k=g_k(\tilde Z_1,...,\tilde Z_k)\ \mathbb P$-a.s.\\
\textbf{Third step:}
Define the Brownian martingale
$\hat M_t=\mathbb E_{\mathbb P}(M_n\,|\,\mathcal{G}_t)$, $t\in [0,T].$
Due to the independent increments of Brownian motion, $\hat M_{\frac{kT}{n}}=M_k$ for any $k$. Define the random variable
$\mathbf X=\max_{0\leq k< n}|M_{k+1}-M_k|$. Now,
let $k< n$ and $t\in [kT/n, (k+1)T/n]$. From Jensen's inequality
$|\hat M_t-\hat M_{kT/n}|\leq \mathbb E_{\mathbb P}(\mathbf X\,|\,\mathcal G_t)$. Thus,
applying Doob's martingale inequality and (\ref{A.1}) yield
\begin{align}
& \ \mathbb P(\max_{0\leq k<n}\max_{kT/n\leq t\leq (k+1)T/n}|\hat M_t-\hat M_{kT/n}|>\sqrt\epsilon) \label{A.2}\\
\leq & \
\mathbb P\left(\max_{0\leq t\leq T}\mathbb E_{\mathbb P}(\mathbf X|\mathcal G_t)>\sqrt\epsilon\right) \nonumber\\
\leq & \
\frac{1}{\sqrt\epsilon}\mathbb E_{\mathbb P} \mathbf X \nonumber\\
= & \
\frac{1}{\sqrt\epsilon}\mathbb E_{Q}\left(\max_{0\leq k<n} |{\mathbb S}_{(k+1)T/n}-{\mathbb S}_{kT/n}|\right) \nonumber\\
\leq & \
\sqrt\epsilon.\nonumber
\end{align}
For $k<n$ and
$\frac{kT}{n}\leq t\leq \frac{(k+1)T}{n}$, we obtain from the Markov property of Brownian motion that
$\hat M_t=\psi_k(\tilde Z_1,...,\tilde Z_{k},t,W_t-W_{kT/n})$, where
$$\psi_k(\tilde Z_1,...,\tilde Z_{k},t,y)=\int_{-\infty}^{\infty}g_{k+1}(\tilde Z_1,...,\tilde Z_{k},v+y)
\frac{e^{-\frac{v^2}{(2k+2)T/n-2t}}
}{\sqrt{2\pi ((k+1)T/n-t)}}dv.$$
From (\ref{A.1+}), we see that the function $g_{k+1}(y_1,...,y_{k+1})$ is non-decreasing in $y_{k+1}$. Hence
the function $\psi_k(\tilde Z_1,...,\tilde Z_{k},t,y)$ is non-decreasing in $y$.
By It\^o's formula, 
$\hat M_t=S_0+\int_{0}^t \beta_u dW_u$, $t\in [0,T]$, with
$\beta_t=\frac{\partial \psi_{[nt/T]}(\tilde Z_1,...,
\tilde Z_{[nt/T]},t,y)}{\partial y}|{\{y=W_t-W_{[nt/T]T/n}\}}$,
$t\in [0,T]$, being a non--negative process.
Finally, set $\alpha_t=\frac{\beta_t}{\hat M_t}$. Then, by construction, $\alpha\in\mathcal{A}$, where $\mathcal{A}$ is the set defined in Section~\ref{sec:lemmas}, which means
 $\alpha$ is a non-negative $\{\mathcal{G}_t\}_{t=0}^T$-progressive process such that
$$ \hat M_t = S_0\, e^{\int_0^t \alpha_v\, dW_v -\frac{1}{2}\int_0^t \alpha^2_v \,dv},\quad t \in [0,T],$$
satisfies $\frac{1}{C}\leq \hat M \leq C$, where the last inequality follows from
(\ref{bb}).
\\
\textbf{Fourth Step:}
Consider the space of all probability measures on $C[0,T]$. Recall the  L\'evy–-Prokhorov metric
 $$d(P_1,P_2)=\inf\{\delta>0: P_1(\mathbb A)\leq \delta+ P_2(\mathbb A^\delta) \ \mbox{and} \
  P_2(\mathbb A)\leq \delta+ P_1(\mathbb A^\delta) \ \forall \mathbb A\},$$ where
$\mathbb A^\delta$ is the set of all function that their distance (in the uniform metric)
 to the set $A$ is smaller than $\delta$. As $C[0,T]$ is a Polish  space, the  L\'evy–-Prokhorov metric induces the topology of weak convergence.
 Define the linear extrapolations
 \begin{align*}
\tilde{\mathbb S}_t&:=
\left(\left[{nt/T}\right]+1-{nt/T}\right){\mathbb S}_{[nt/T]T/n}+
\left(nt/T-\left[nt/T\right]\right){\mathbb S}_{([nt/T]+1)T/n},
 \quad t\leq T,\\
 \tilde{M}_t&:=
 \left(\left[{nt/T}\right]+1-{nt/T}\right)M_{[nt/T]}+
 \left(nt/T-\left[nt/T\right]\right)M_{[nt/T]+1},
  \quad t\leq T.
 \end{align*}
As a consequence of the second step, we obtain that the distribution of $\tilde {\mathbb S}$ (under $Q$)
equals the distribution of $\tilde M$ (under $\mathbb P)$. Denote it by $Q_1$.
From (\ref{A.1}) and the Markov inequality
we obtain that
$d(Q,Q_1)\leq \sqrt\epsilon.$
The inequality (\ref{A.2}) implies that
$d( Q_1, Q_2)\leq 2\sqrt\epsilon$ where $Q_2$ is the distribution of $\hat M$ (under $\mathbb P$).
Thus, we get $d(Q,Q_2)\leq 3\sqrt\epsilon$.
As $\epsilon>0$ was arbitrary, we obtain that the set of distributions of
$\mathbf S^{(\alpha)}$ (recall the definition after formula (\ref{3.21})),
$\alpha\in\mathcal A$,
is dense in $\mathcal M^b$, and in view of the first step we obtain that
the set of distributions of $\mathbf S^{(\alpha)}$,
$\alpha\in\mathcal A$, is dense in $\mathcal M$.
 Moreover, using similar arguments
as in Lemma~\ref{lem4.1}, we conclude that the set of distributions
of $\mathbf S^{(\alpha)}$, $\alpha\in\mathcal C(\nu_0)$, is dense in $\mathcal M$.
We arrive to the final step.
${}$\\
\textbf{Fifth step:}
From the last step, it follows that it is sufficient to prove that, for
any $\alpha\in\mathcal C(\nu_0)$, the distribution of $\mathbf S^{(\alpha)}$ lies in the weak closure of
$\mathcal M^S$. Thus, choose $\alpha\in\mathcal C(\nu_0)$.
We use the property of fully incomplete market. By Definition~\ref{dfn2.1}, there exists a sequence of probability measures
$\mathbb Q_n\ll \mathbb P$, $n\in\mathbb N$, such that (\ref{dnew}) holds for $\epsilon=\frac{1}{n}$
and $W$ is a $\mathbb Q_n$ Brownian motion. As $\alpha$ is adapted to $\{\mathcal G_t\}_{t=0}^T$
then the distribution of $(\alpha,W)$ under $\mathbb Q_n$ is the same as under $\mathbb P$. Hence,
the distribution of $(\nu,W)$ under $\mathbb Q_n$
converges weakly as $n\rightarrow\infty$ (on the space $C[0,T]\times C[0,T]$) to the distribution
of $(\alpha,W)$ under $\mathbb P$.
Recall that
\begin{align*}
d\tilde S_t&=S_0+\int_{0}^t \nu_t\tilde S_t \,dW_t, \ \ t\in [0,T], \ \ \mathbb Q_n\mbox{-a.s.,}\\
d\mathbf S^{(\alpha)}_t&=S_0+\int_{0}^t \alpha_t\mathbf{S}^{(\alpha)}_t dW_t, \ \ t\in [0,T] \ \ \mathbb {P}\mbox{-a.s}.
\end{align*}
Thus, from  Duffie and Protter (1992, Proposition 4.1 and Theorem 4.3--4.4), we obtain
that the distribution of $\tilde S$ under $\mathbb Q_n$
converges weakly to the distribution of $\mathbf S^{(\alpha)}$, as required.
\end{proof}

\begin{rem}
It is possible to define a fully incomplete market
as a market which satisfies that the set of distributions
$$\{\mathbb Q(S\in\cdot): \mathbb Q \ \mbox{is} \ \mbox{an} \
\mbox{equivalent} \ \mbox{martingale} \ \mbox{measure}\}$$
is a weakly dense subset of $\mathcal M$.
This is the only property that we used in
the proof of Theorem~\ref{thm2.1}.
However, when dealing with game options
(or any options which involve stopping times)
such as Theorem~\ref{thm2.1game}, we
need an additional structure related to the
filtration $\{\mathcal F_t\}_{t=0}^T$. This additional structure
is given by (\ref{2.2}) and Definition \ref{dfn2.1}.
\end{rem}


\end{document}